\documentclass[12pt, draftclsnofoot, onecolumn]{IEEEtran}
\usepackage{amsfonts}
\usepackage{mathrsfs}
\usepackage{bbm}
\usepackage[ruled,linesnumbered]{algorithm2e}
\usepackage{paralist}
\usepackage{algpseudocode}
\usepackage{multicol}
\usepackage{stfloats}
\usepackage{bm}
\usepackage{xcolor}
\usepackage{cite}
\usepackage{graphicx}
\usepackage{float}
\usepackage{subfigure}
\usepackage{stfloats}
\usepackage{ntheorem}
\usepackage{amsmath}
\usepackage{amssymb}
\usepackage{flushend,cuted}
\theoremseparator{:}
\theoremheaderfont{\it}
\theorembodyfont{\upshape}
\newtheorem{thm}{Theorem}
\newtheorem{lem}{Lemma}

\newtheorem{rem}{Remark}
\newtheorem{defi}{Definition}
\newtheorem*{proof}{Proof}
\newtheorem{coro}{Corollary}
\SetKwInput{Require}{Require}
\SetKwInput{Ensure}{Ensure}
\begin{document}
\newcommand{\e}[1]{\boldsymbol{#1}}

\title{Privacy-preserving Channel Estimation in Cell-free Hybrid Massive MIMO Systems}

\IEEEoverridecommandlockouts

\author{Jun~Xu,
        Xiaodong~Wang,~\IEEEmembership{Fellow,~IEEE},
        Pengcheng~Zhu, ~\IEEEmembership{Member,~IEEE},
        and~Xiaohu~You,~\IEEEmembership{Fellow,~IEEE}\\
\thanks{J. Xu, P. Zhu and X. You are all with National Mobile Communications Research Laboratory, Southeast University, Nanjing, China ( emails: \{xujunseu, p.zhu, xhyu\}@seu.edu.cn).}
\thanks{X.Wang is with the Department of Electrical Engineering, Columbia University, New York, NY 10027 (e-mail:wangx@ee.columbia.edu).}
}

\maketitle

\begin{abstract}
We consider a cell-free hybrid massive multiple-input multiple-output (MIMO) system with $K$ users and $M$ access points (APs), each with $N_a$ antennas and $N_r< N_a$ radio frequency (RF) chains. When $K\ll M{N_a}$, efficient uplink channel estimation and data detection with reduced number of pilots can be performed based on low-rank matrix completion. However, such a scheme requires the central processing unit (CPU) to collect received signals from all APs, which may enable the CPU to infer the private information of user locations. We therefore develop and analyze privacy-preserving channel estimation schemes under the framework of differential privacy (DP). As the key ingredient of the channel estimator, two joint differentially private noisy matrix completion algorithms based respectively on Frank-Wolfe iteration and singular value decomposition are presented. We provide an analysis on the tradeoff between the privacy and the channel estimation error. In particular, we show that the estimation error can be mitigated while maintaining the same privacy level by increasing the payload size with fixed pilot size; and the scaling laws of both the privacy-induced and privacy-independent error components in terms of payload size are characterized. Simulation results are provided to further demonstrate the tradeoff between privacy and channel estimation performance. 
\end{abstract}

\section{Introduction}
Due to the high spectral and energy efficiencies, the cell-free massive MIMO has emerged as a promising wireless technology, where a large number of access points (APs) are distributed over a geographical area, and collaboratively serve users using the same time-frequency resource\cite{XHWang20IOT}. To reduce the high cost associated with equipping each antenna with a radio frequency (RF) chain that contains a high-resolution analog-to-digital converter (ADC)\cite{SHanCM15}, hybrid analog/digital architectures are typically employed where with analog combining based on switches or phase shifters, antennas are randomly connected to a reduced number of RF chains and ADCs\cite{SLiangTWC19}.

To enable cell-free hybrid massive MIMO systems, it is crucial to obtain accurate channel state information (CSI). In \cite{SLiangTWC19}, a semi-blind channel estimation method based on low-rank matrix completion was proposed for hybrid massive MIMO, with the salient feature that the number of pilots is proportional to the number of users, instead of the number of antennas; and the estimation error reduces with the increase of the data payload size. In order to apply such channel estimation scheme in a cell-free system, each AP needs to send its observed received signal to a central processing unit (CPU), which performs channel estimation and data detection for all users. However, this may lead to the leakage of users' location information to the CPU, since the large scale fading of channels are determined by the locations of users and APs according to the path loss law. 

Nowadays, the privacy awareness of the public has been significantly increased when using smart mobile devices and services. From the view point of users, privacy in 5G network can be divided into three
main categories: data privacy, location privacy and identity privacy\cite{Liyanage}. Locations are usually
regarded as one of the most important sensitive information for most people, the leakage of which may pose
threats to other sensitive information (e.g., home address, work place) and even personal safety\cite{WTong}. Hence, it is crucial to provide high-quality services without disclosing the users’ location privacy in 5G mobile networks.

Differential privacy (DP) is a probabilistic framework based on the notion of indistinguishability\cite{Dwork10}. In particular, observing an output of a differentially private algorithm, one cannot infer whether any specific user contributed to the data. In this framework, privacy is mainly achieved by randomizing the released statistics. DP has been accepted as a standard privacy model and widely adopted in many fields, such as recommender system\cite{mcsherry2009differentially}, deep learning\cite{abadi2016deep}, distributed optimization\cite{huang2015differentially}, data mining\cite{mohammed2011differentially}, ridesharing services\cite{WTong}, etc. In addition, applications of DP in communication networks include data-driven caching in information-centric networks\cite{XZhang18}, and big data analytics in edge computing\cite{MDuCM18,CXuCM18}. 
However, so far there is no work addressing DP for physical-layer signal processing. A particular challenge is that unlike the above-mentioned higher-layer applications, the physical-layer is much more sensitive to the perturbation noise added to achieve DP.

In this paper, we aims to design privacy-preserving channel estimation algorithms for cell-free hybrid massive MIMO systems. The major contributions are summarized as follows:
\begin{itemize}
    \item To the best of our knowledge, this is the first work that integrates DP with physical-layer signal processing. 
    \item We propose two privacy-preserving channel estimators based on Frank-Wolfe (FW) iteration and singular value decomposition (SVD), respectively. 
    \item We show that both channel estimation algorithms are joint differentially private. We also analyze the estimation error bounds for the two algorithms, and characterize the scaling laws of the estimation error in terms of data payload size. 
    \item Through extensive simulations, we illustrate the tradeoff between privacy and channel estimation and data detection performance for the two algorithms. 
\end{itemize}

The remainder of this paper is organized as follows. Section \uppercase\expandafter{\romannumeral2} describes the cell-free hybrid massive MIMO system under consideration and provides some background on DP. Two privacy-preserving channel estimation algorithms are proposed in Section \uppercase\expandafter{\romannumeral3}. Section \uppercase\expandafter{\romannumeral4} presents the analysis on the privacy and channel estimation performance of the two algorithms. Simulation results are provided in Section \uppercase\expandafter{\romannumeral5}. Finally, Section \uppercase\expandafter{\romannumeral6} concludes the paper.

\emph{Notations:} Boldface letters denote matrices (upper case) or vectors (lower case). The transpose, conjugate transpose and trace operators are denoted by ${\left(  \cdot  \right)^ \textrm{T} }$, ${\left(  \cdot  \right)^ \textrm{H} }$ and ${\rm{tr}}\left(  \cdot  \right)$ respectively. $\|\cdot\|_{\mathcal{F}}$, $\|\cdot\|_{2}$ and $\|\cdot\|_{\rm{nuc}}$ denote the Frobenius norm, spectral norm and nuclear norm of a matrix, respectively. Assuming the singular values of a matrix $\mathbf{A}\in\mathbb{C}^{m\times n}$ are $\lambda_1,\cdots, \lambda_{\min\left(m,n\right)}$ in descending order, then we have $\|\mathbf{A}\|_{\mathcal{F}}=\sqrt{\sum\nolimits_{i=1}^m\sum\nolimits_{j=1}^n{\mathbf{A}^2\left(i,j\right)}}=\sqrt{\sum\nolimits_{i=1}^{\min\left(m,n\right)}\lambda_i^2}$; $\|\mathbf{A}\|_{2}=\lambda_1$; $\|\mathbf{A}\|_{\rm{nuc}}=\sum\nolimits_{i=1}^{\min\left(m,n\right)}{\lambda_i}$. ${\rm{Diag}}\left(\mathbf{d}\right)$ returns a diagonal matrix whose diagonal elements are given by a vector $\mathbf{d}$. ${{\bf{I}}_M}$, $\otimes$ and $\mathbb{E}\left\{  \cdot  \right\}$ respectively represent the $M\times M$ identity
matrix, the Kronecker product and the expectation operator. $\mathcal{\mathcal{N}_{\rm{c}}}(\mu,\sigma^2)$ and $\mathcal{\mathcal{N}}(\mu,\sigma^2)$ respectively denote the complex and real circularly symmetric Gaussian distribution with mean $\mu$
and variance $\sigma^2$. $f(n)=\Theta\left(g(n)\right)$ means $f$ is bounded below by $g$ asymptotically; $f(n)=O\left(g(n)\right)$ means $f$ is bounded above by $g$ asymptotically; $f(n)=\omega\left(g(n)\right)$ means $f$ dominates $g$ asymptotically.

\section{System Descriptions and Background}
\subsection{Signal Model}
We consider a cell-free massive MIMO system, in which $M$ distributed APs each equipped with $N_a$ antennas collaboratively serve $K$ single-antenna users using the same time-frequency resource, as shown in Fig. \ref{fig:system}. We denote ${\cal M} = \left\{ {1, \ldots M} \right\}$ and ${\cal K} = \left\{ {1, \ldots K} \right\}$ as the sets of APs and users respectively. Each AP employs an analog structure with $N_r$ RF chains to combine the incoming signal in the RF band. Each RF chain contains a high-resolution ADC and forwards the data stream to the baseband processor that performs only simple signal processing.  All APs are connected to a CPU through perfect backhaul links, which performs computation-intensive signal processing. 

\begin{figure}
    \centering
    \includegraphics[width=1\textwidth]{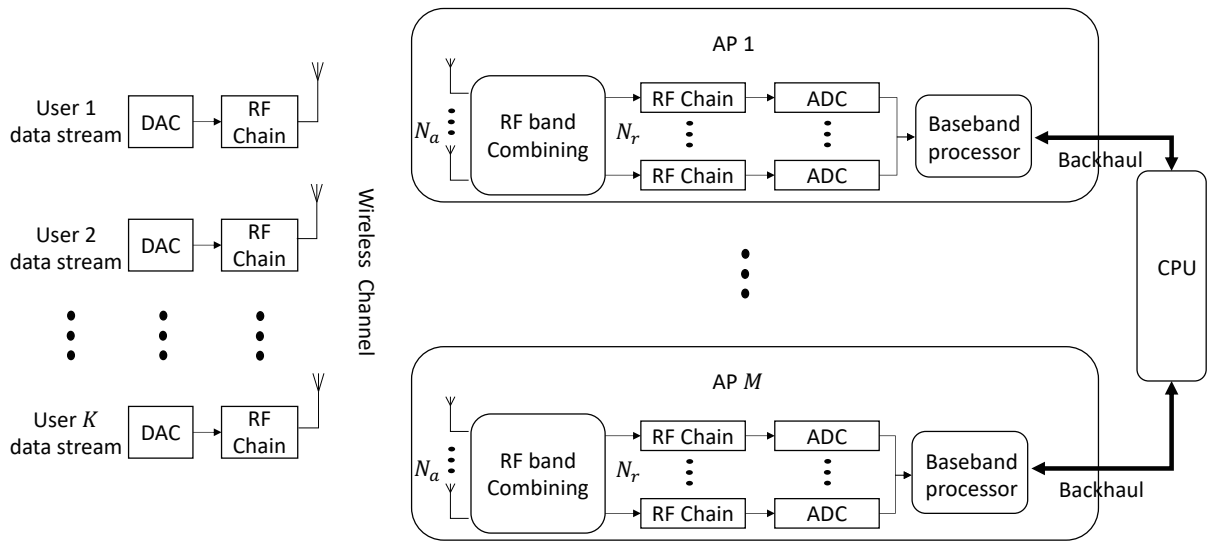}
    \caption{The cell-free multiuser massive MIMO uplink system. Left: $K$ single-antenna users. Right: $M$ APs each with a hybrid structure consisting of $N_a$ antennas and $N_r$ RF chains. Each AP has a baseband processor with limited computational capability and all APs are connected to a CPU via perfect backhual links. }
    \label{fig:system}
\end{figure}

We assume a block flat-fading channel between each user-AP pair. Let $\mathcal{T}_{\rm{c}}=\left\{1,\cdots,\tau_{\rm{c}}\right\}$ denote the set of time slots within a coherence interval $\tau_{\rm{c}}$ when the channel coefficients remain constant. Throughout the paper, we assume that $M{N_a}>\tau_{\rm{c}}$, which can be satisfied in massive MIMO. The first $
\tau_{\rm{p}}$ time slots denoted by $\mathcal{T}_{\rm{p}}=\left\{1,\cdots,\tau_{\rm{p}}\right\}$ are used for uplink channel estimation, and the remaining $\tau_{\rm{d}}=\tau_{\rm{c}}-\tau_{\rm{p}}$ time slots denoted by $\mathcal{T}_{\rm{d}}=\left\{\tau_{\rm{p+1}},\cdots,\tau_{\rm{c}}\right\}$ are used for uplink data transmission. The channel vector from user $k$ to AP $m$ can be modeled as
\begin{equation}
\label{equ_channel}
    {{\mathbf{h}}_{k,m}}=\sqrt{\beta_{k,m}}{{\mathbf{g}}_{k,m}}\in {{\mathbb{C}}^{N_a\times {1}  }},
\end{equation}
where $\beta_{k,m}$ represents the large-scale fading and ${{\mathbf{g}}_{k,m}}\in{{\cal{N}}_{\rm{c}}}\left( {0,\mathbf{I}_{N_a}} \right)$ models the small-scale fast fading. 

Let $\mathbf{s}\left[t\right]\sim \mathcal{N}_{\rm{c}}\left(0,\mathbf{I}_{K}\right)$ denote the transmitted signal from $K$ users at time slot $t$, i.e., $\mathbf{s}\left[t\right]$ corresponds to pilots for $t\in{\mathcal{T}}_{\rm{p}}$ and data symbols for $t\in{\mathcal{T}}_{\rm{d}}$. The received signal $\mathbf{r}_m\left[t\right]\in\mathbb{C}^{{N_a}\times 1}$ across ${N_a}$ antennas at AP $m$ is given by
\begin{equation}
\label{equ_rm}
 \mathbf{r}_m\left[t\right]=\mathbf{H}_m\mathbf{s}\left[t\right]+\mathbf{n}_m\left[t\right], \forall t\in \mathcal{T}_{\rm{c}},\forall m\in\mathcal{M},
\end{equation}
where $\mathbf{H}_m=\left[\mathbf{h}_{1,m},\cdots,\mathbf{h}_{K,m}\right]\in\mathbb{C}^{{N_a}\times K}$ denotes the channel matrix between AP $m$ and all users. $\mathbf{n}_m\left[t\right]\sim\mathcal{N}_{\rm{c}}\left(0,\sigma^2{\mathbf{I}_{{N_{a}}}}\right)$ is the received noise sample at AP $m$ at time slot $t$. Denote $\mathbf{R}_m\buildrel\Delta\over=\left[\mathbf{r}_m\left[1\right],\cdots,\mathbf{r}_m\left[\tau_{\rm{c}}\right]\right]\in\mathbb{C}^{{N_a}\times \tau_{\rm{c}}}$, $\mathbf{N}_m\buildrel\Delta\over=\left[\mathbf{n}_m\left[1\right],\cdots,\mathbf{n}_m\left[\tau_{\rm{c}}\right]\right]\in\mathbb{C}^{{N_a}\times \tau_{\rm{c}}}$, $\mathbf{P}\buildrel\Delta\over=\left[\mathbf{s}\left[1\right],\cdots,\mathbf{s}\left[\tau_{\rm{p}}\right]\right]\in\mathbb{C}^{K\times \tau_{\rm{p}}}$, $\mathbf{D}\buildrel\Delta\over=\left[\mathbf{s}\left[\tau_{\rm{p}}+1\right],\cdots,\mathbf{s}\left[\tau_{\rm{c}}\right]\right]\in\mathbb{C}^{K\times \tau_{\rm{d}}}$, $\mathbf{S}= \left[\mathbf{P},\mathbf{D}\right]\in\mathbb{C}^{K\times{\tau_{\rm{c}}}}$. Then (\ref{equ_rm}) can be rewritten as
\begin{equation}
\label{equ_Rm}
    \mathbf{R}_m=\mathbf{H}_m\mathbf{S}+\mathbf{N}_m,\forall m\in\mathcal{M}.
\end{equation}
By stacking the signals from all APs and denoting $\mathbf{R}=\left[\mathbf{R}_1^{\rm{T}},\cdots,\mathbf{R}_M^{\rm{T}}\right]^{\rm{T}}\in\mathbb{C}^{M{N_a}\times{\tau_{\rm{c}}}}$, $\mathbf{H}=\left[\mathbf{H}_1^{\rm{T}},\cdots,\mathbf{H}_M^{\rm{T}}\right]^{\rm{T}}\in\mathbb{C}^{M{N_a}\times K}$ and $\mathbf{N}=\left[\mathbf{N}_1^{\rm{T}},\cdots,\mathbf{N}_M^{\rm{T}}\right]^{\rm{T}}\in\mathbb{C}^{M{N_a}\times{\tau_{\rm{c}}}}$, we then have
\begin{equation}
\label{equ_R}
    \mathbf{R}=\mathbf{H}\mathbf{S}+\mathbf{N}.
\end{equation}
Note that ${\rm{rank}}\left(\mathbf{H}\mathbf{S}\right)\le K$ and we assume that $M{N_a}\gg K$ and $\tau_{\rm{c}}\gg K$, hence $\mathbf{R}$ is a noisy version of a low-rank matrix.

Then $\mathbf{r}_m\left[t\right]$ will pass through analog structures and be combined in the RF band. In this paper, we consider the analog combining based on switches, where each RF chain is randomly connected to one of $N_a$ antennas through a switch at each time slot\cite{SLiangTWC19}. Such antenna selection can capture many advantages of massive MIMO and has low power consumption\cite{SSanayeiCM2004}. 
Denote $\Omega_m$ as the set of indices $(n,t)$ such that the $n$-th element of ${\mathbf{r}}_m\left[t\right]$, ${\mathbf{R}}_m(n,t)$ is observed. Denote ${\mathbf{Y}}_m=(\mathbf{R}_m)_{\Omega_m}\in\mathbb{C}^{{N_a}\times \tau_{\rm{c}}}$ as the sampled version of ${\mathbf{R}}_m$ such that 
\begin{equation}
\label{equ_YmR}
    \mathbf{Y}_m\left(n,t\right)=\left\{\begin{matrix}
{{\mathbf{R}}_m}\left(n,t\right), & {\rm{if}} \left(n,t\right)\in\Omega_m,\\ 
 0, & {\rm{if}} \left(n,t\right)\notin\Omega_m.
\end{matrix}\right.    
\end{equation}
Note that each column of $\mathbf{Y}_m$ has exactly $N_r$ non-zero elements.

Traditionally when there is no privacy concern, to fully exploit the low-rank structure  of 
$\mathbf{R}$ in (4) and the computing power of CPU, each AP $m$ will send its received signal $\mathbf{Y}_m$ to the CPU. The CPU will then estimate the channels $\left\{\mathbf{H}_m, m \in \mathcal{M} \right\}$ and the user data payload $\mathbf{D}$, based on 
$\left\{\mathbf{Y}_m, m \in \mathcal{M}\right\}$ and the pilots $\mathbf{P}$.  
However, note that the large-scale fading coefficient $\beta_{k,m}$ in (\ref{equ_channel}) contains the path loss information which is in turn determined by the distance between the user-AP pair. Since the locations of APs are fixed, the location of user $k$ can be accurately estimated if its distances from more than three APs are known. Hence if each AP $m$ directly sends $\mathbf{Y}_m$ to the CPU, then the location information of users
might be compromised once the CPU obtains accurate estimates of the channels $\left\{\mathbf{H}_m, m \in \mathcal{M} \right\}$.
Hence, in order to protect the location privacy of users, each AP $m$ cannot send its $\mathbf{Y}_m$ directly to the CPU. Instead, the APs and the CPU should collaborate in a privacy-preserving way such that the estimate of channel $\mathbf{H}_m$ is only available to AP $m$ but not to the CPU or other APs. 

This paper focuses on the design and analysis of such privacy-preserving channel estimation schemes. In the next subsection, we provide a general overview of the notion of differential privacy (DP) and basic approaches to achieving DP.

\subsection{Background on Differential Privacy (DP)}
Recall that $\mathbf{Y}_m$ and $\mathbf{H}_m$ denote the received signal and the channel at AP $m$, respectively. Denote $\widehat{\mathbf{H}}_m$ as an estimate of $\mathbf{H}_m$, $\mathbf{Y}=\left[\mathbf{Y}_1^{\rm{T}},\cdots,\mathbf{Y}_M^{\rm{T}}\right]^{\rm{T}}\in\mathbb{C}^{M{N_a}\times {\tau_{\rm{c}}}}$ and $\mathbf{\widehat{H}}=\left[\mathbf{\widehat{H}}_1^{\rm{T}},\cdots,\mathbf{\widehat{H}}_M^{\rm{T}}\right]^{\rm{T}}\in\mathbb{C}^{M{N_a}\times K}$. Let $\mathcal{A}: \mathbb{C}^{M{N_a}\times {\tau_{\rm{c}}}}\rightarrow{\mathbb{C}^{M{N_a}\times K}}$ be a randomized channel estimation algorithm which takes $\mathbf{Y}$ as input, and outputs $\mathbf{\widehat{H}}$. In addition, $\mathcal{A}_{-m}: \mathbb{C}^{M{N_a}\times {\tau_{\rm{c}}}}\rightarrow\mathbb{C}^{(M-1){N_a}\times K}$ outputs $\mathbf{\widehat{H}}_{-m}$, which denotes the estimated channels for all APs other than AP $m$. Similarly denote $\mathbf{Y}_{-m}\in\mathbb{C}^{(M-1){N_a}\times {\tau_{\rm{c}}}}$ as $\mathbf{Y}$ with $\mathbf{Y}_{m}$ removed. Recall that our goal is to devise the channel estimation algorithm ${\mathcal{A}}$ such that 
$\widehat{\mathbf{H}}_m$ is available only to AP $m$, but not to the CPU or other APs. 
\begin{defi}[Standard DP\cite{dwork2006calibrating} and Joint DP\cite{kearns2014mechanism}]
Given $\epsilon, \delta>0$, $\mathcal{A}$ is $\left(\epsilon,\delta\right)$-differentially private if for any AP $m$, any two possible values $\mathbf{Y}_m,\mathbf{Y}_m^{'}$ of the received signal at AP $m$, any value $\mathbf{Y}_{-m}$ of the received signal at all other APs, and any subset $S\in\mathbb{C}^{M{N_a}\times K}$, we have
\begin{equation}
    {P\left[{\mathcal{A}}\left(\left[\mathbf{Y}_m,\mathbf{Y}_{-m}\right]\right)\in{S}\right]}\le{e^{\epsilon}}{P\left[{\mathcal{A}}\left(\left[\mathbf{Y}_m^{'},\mathbf{Y}_{-m}\right]\right)\in{S}\right]}+\delta,
\end{equation}
where the probability $P\left[\cdot\right]$ is over the randomness of $\mathcal{A}$. Moreover, $\mathcal{A}$ is $\left(\epsilon,\delta\right)$-joint differentially private if for any subset $S\in\mathbb{C}^{(M-1){N_a}\times K}$, we have
\begin{equation}
    {P\left[{\mathcal{A}}_{-m}\left(\left[\mathbf{Y}_m,\mathbf{Y}_{-m}\right]\right)\in{S}\right]}\le{e^{\epsilon}}{P\left[{\mathcal{A}}_{-m}\left(\left[\mathbf{Y}_m^{'},\mathbf{Y}_{-m}\right]\right)\in{S}\right]}+\delta.
\end{equation}
\end{defi}

The meaning of the above standard DP is that a change of the received signal $\mathbf{Y}_m$ at any AP $m$ has a negligible impact on the estimated channel $\mathbf{\widehat{H}}$. Hence, one cannot infer much about the private data $\mathbf{Y}_m$ from the output $\mathbf{\widehat{H}}$ and the data $\mathbf{Y}_{-m}$, which however, implies that the estimated $\mathbf{\widehat{H}}_m$ of AP $m$ should not depend strongly on $\mathbf{Y}_m$. Obviously such channel estimator will not be of practical value. On the other hand, joint DP means that the estimate $\mathbf{\widehat{H}}_m$ for any particular AP $m$ can depend strongly on its data $\mathbf{Y}_m$; but the output of all other APs $\mathbf{\widehat{H}}_{-m}$ and $\mathbf{Y}_{-m}$ do not reveal much about $\mathbf{Y}_m$, which is a meaningful notion of privacy in the context of channel estimation and will be adopted hereafter. Here, $\epsilon$ and $\delta$ represent the worst-case privacy loss and smaller values of them imply stronger privacy guarantee.

Next we review two important properties of DP, which hold for both standard DP and joint DP.
\begin{lem}[Post-Processing\cite{dwork2014algorithmic}]
\label{lem_PP}
Let $\mathcal{A}: \mathcal{T}\rightarrow{\mathcal{S}
}$ and $\mathcal{B}: \mathcal{S}\rightarrow{\mathcal{Q}
}$ be randomized algorithms. Define an algorithm $\mathcal{A}^{'}: \mathcal{T}\rightarrow{\mathcal{Q}
}$ by $\mathcal{A}^{'}=\mathcal{B}\left(\mathcal{A}\right)$. If $\mathcal{A}$ satisfies $\left(\epsilon,\delta\right)$-(joint) DP, then $\mathcal{A}^{'}$ also satisfies $\left(\epsilon,\delta\right)$-(joint) DP.
\end{lem}

Hence any operation performed on the output of a (joint) differentially private algorithm, without accessing the raw data, remains (joint) differentially private with the same level of privacy.

\begin{lem}[$T$-fold composition\cite{Dwork10,dwork2014algorithmic}]
\label{lem_TFC}
We assume that there are $T$ independent randomized algorithms $\mathcal{A}_1,\cdots, \mathcal{A}_T$, where each algorithm $\mathcal{A}_t:\mathcal{D}\rightarrow{\mathcal{R}_t}$ is $\left(\epsilon,\delta\right)$-(joint) differentially private. For all $\epsilon, \delta, \delta'>0$, an algorithm $\mathcal{A}$ defined as $\mathcal{A}\left(\mathcal{D}\right)=\left(\mathcal{A}_1\left(\mathcal{D}\right),\cdots,\mathcal{A}_T\left(\mathcal{D}\right)\right)$ satisfies $\left(\epsilon',T\delta+\delta'\right)$-(joint) DP with
\begin{equation}
    \epsilon'=\epsilon\sqrt{2T\ln{\left(1/\delta'\right)}}+T\epsilon\left(e^{\epsilon}-1\right).
\end{equation}
In particular, given target privacy parameters $0<\epsilon'< 1$ and $\delta'>0$, $\mathcal{A}$ satisfies $\left(\epsilon',T\delta+\delta'\right)$-(joint) DP if each algorithm is $\left(\epsilon'/\sqrt{8T\ln{\left(1/\delta'\right)}},\delta\right)$-(joint) differentially private.
\end{lem}

In the context of channel estimation, if $\mathbf{Y}$ is accessed by CPU $T$ times, each denoted by $\mathcal{A}_t\left(\mathbf{Y}\right), t=1,\cdots, T$, then the information released to CPU is $\mathcal{A}\left(\mathbf{Y}\right)=\left(\mathcal{A}_1\left(\mathbf{Y}\right),\cdots,\mathcal{A}_T\left(\mathbf{Y}\right)\right)$. To make $\mathcal{A}\left(\mathbf{Y}\right)$ joint differentially private,
we need to guarantee that each access $\mathcal{A}_t\left(\mathbf{Y}\right)$ is joint differentially private. In addition, the difference in the privacy level between $\mathcal{A}\left(\mathbf{Y}\right)$ and $\mathcal{A}_t\left(\mathbf{Y}\right)$ stated in this lemma will be useful in the design of the privacy-preserving channel estimator.

The following important lemma provides us a way to achieve joint DP by standard DP.
\begin{lem}[Billboard Lemma\cite{hsu2016private}]
Suppose $\mathcal{A}: \mathcal{D}=\left(\mathcal{D}_1,\cdots,\mathcal{D}_M\right)\rightarrow{\mathcal{R}
}$ is $\left(\epsilon,\delta\right)$-differentially private, where $\mathcal{D}_m$ denotes the data of AP $m$. If a randomized algorithm $\mathcal{B}$ has $M$ components with the $m$-th component $\mathcal{B}_m\left(\mathcal{D}_m, \mathcal{A}\left(\mathcal{D}\right)\right)$, where $\mathcal{B}_m: \mathcal{D}_m \times {\mathcal{R}
}\rightarrow{\mathcal{Q}_m},\forall m\in\mathcal{M}$, then $\mathcal{B}$ is $\left(\epsilon,\delta\right)$-joint differentially private.
\end{lem}

Next we review the definition of ${\ell _2}$-sensitivity and a well-known approach to achieving standard DP.
\begin{defi}[${\ell _2}$-Sensitivity\cite{dwork2014algorithmic}]
\label{def_sen}
Let $f\left(\mathbf{Y}\right)$ be an arbitrary function on the received signal $\mathbf{Y}=\left[\mathbf{Y}_1^{\rm{T}},\cdots,\mathbf{Y}_M^{\rm{T}}\right]^{\rm{T}}$. Then its ${\ell _2}$-sensitivity is defined as the maximum difference in the function values when the received signals differ only at one AP, i.e., 
\begin{equation}
    \Delta_f\buildrel \Delta \over =\max_{1\leq m\leq M}\max_{\mathbf{Y}_m\neq\mathbf{Y}_m^{'}\atop\mathbf{Y}_i=\mathbf{Y}_i^{'},\forall i\neq m}\left\|f\left(\mathbf{Y}\right)-f\left(\mathbf{Y}^{'}\right)\right\|_{\mathcal{F}}.
\end{equation}
\end{defi}

\begin{lem}[Gaussian Mechanism\cite{dwork2014algorithmic}]
\label{lem_GM}
Assuming that the information released to CPU during channel estimation is
\begin{equation}
{\mathcal{A}}\left(\mathbf{Y}\right)=f\left(\mathbf{Y}\right)+\mathbf{G},
\end{equation}
where $f\left(\mathbf{Y}\right)$ has the ${\ell _2}$-sensitivity $\Delta_f$ and
\begin{equation}
    \mathbf{G}{(i,j)}\stackrel{\mbox{i.i.d.}} \sim{\mathcal{N}_{\rm{c}}\left(0,\frac{\Delta_f^2{2\ln(1.25/\delta)}}{\epsilon^2}\right)}.
\end{equation}
Then the released ${\mathcal{A}}\left(\mathbf{Y}\right)$ satisfies $\left(\epsilon,\delta\right)$-DP.
\end{lem}

The above lemma helps us to calibrate the Gaussian perturbation noise to achieve $\left(\epsilon,\delta\right)$-DP. It can be seen that larger perturbation noise is required to achieve stronger privacy level, i.e., smaller $\epsilon$ and/or $\delta$, which is intuitive because larger noise variance increases the uncertainty about the released information and hence improves privacy.

\section{Privacy Preserving Channel Estimation}
In this section, we first show that the key component of the privacy-preserving channel estimator is privacy-preserving matrix completion.  We then outline two such matrix completion algorithms. 

\subsection{Channel Estimation Based on Matrix Completion}
Recall that in (\ref{equ_R}), $\mathbf{X}=\mathbf{H}\mathbf{S}\in\mathbb{C}^{{M{N_a}}\times{\tau_{\rm{c}}}}$ is low-rank, i.e., ${\rm{rank}}\left(\mathbf{X}\right)\leq K$, and $\mathbf{Y}$ is an incomplete observation of $\mathbf{X}$ corrupted by channel noise $\mathbf{N}$.
We first design a privacy-preserving algorithm $\mathcal{A}$ to solve a noisy matrix completion problem, which takes $\mathbf{Y}$ as input and outputs a low-rank matrix  $\widehat{\mathbf{X}}=\left[\left(\widehat{\mathbf{X}}_1\right)^{\rm{T}},\cdots,\left(\widehat{\mathbf{X}}_M\right)^{\rm{T}}\right]^{\rm{T}}\in\mathbb{C}^{{M{N_a}}\times{\tau_{\rm{c}}}}$ as an estimate of $\mathbf{X}=\left[\mathbf{X}_1^{\rm{T}},\cdots,\mathbf{X}_M^{\rm{T}}\right]^{\rm{T}}$, where $\mathbf{X}_m=\mathbf{H}_m{\mathbf{S}}$. Then, the estimation of channel $\mathbf{H}_m$ can be performed locally at AP $m$ based on $\widehat{\mathbf{X}}_m$.  

To satisfy $\left(\epsilon,\delta\right)$-joint DP, the matrix completion algorithm $\mathcal{A}$ consists of a global component $\mathcal{A}^{\rm{G}}$ at the CPU and $M$ local components $\mathcal{A}_m^{\rm{L}}=\left(\mathcal{A}_m^{{\rm{L}},1},\mathcal{A}_m^{{\rm{L}},2}\right),m=1,\cdots,M$ at $M$ APs. The local component at AP $m$ first performs a privacy-preserving computation $\mathcal{A}_m^{\rm{L},1}$ on private $\mathbf{Y}_m$ to get $\mathcal{A}_m^{\rm{L},1}\left(\mathbf{Y}_m\right)$, and then transmits it to the CPU through backhaul links. The global component $\mathcal{A}^{\rm{G}}$ aggregates the received information $\mathbf{Z}=\sum\limits_{m=1}^M\mathcal{A}_m^{\rm{L},1}\left(\mathbf{Y}_m\right)$, and computes $\mathcal{A}^{\rm{G}}\left(\mathbf{Z}\right)$, which is then broadcast to all APs through backhaul links. Based on the public $\mathcal{A}^{\rm{G}}\left(\mathbf{Z}\right)$ and private $\mathbf{Y}_m$, the local component at AP $m$ then performs a computation $\mathcal{A}_m^{\rm{L},2}$ to get a complete matrix $\widehat{\mathbf{X}}_m=\mathcal{A}_m^{\rm{L},2}\left(\mathcal{A}^{\rm{G}}\left(\mathbf{Z}\right),\mathbf{Y}_m\right)$. Such a global-local computation model for privacy-preserving noisy matrix completion is depicted in Fig. \ref{fig:compModel}.
\begin{figure}
    \centering
    \includegraphics[width=1\textwidth]{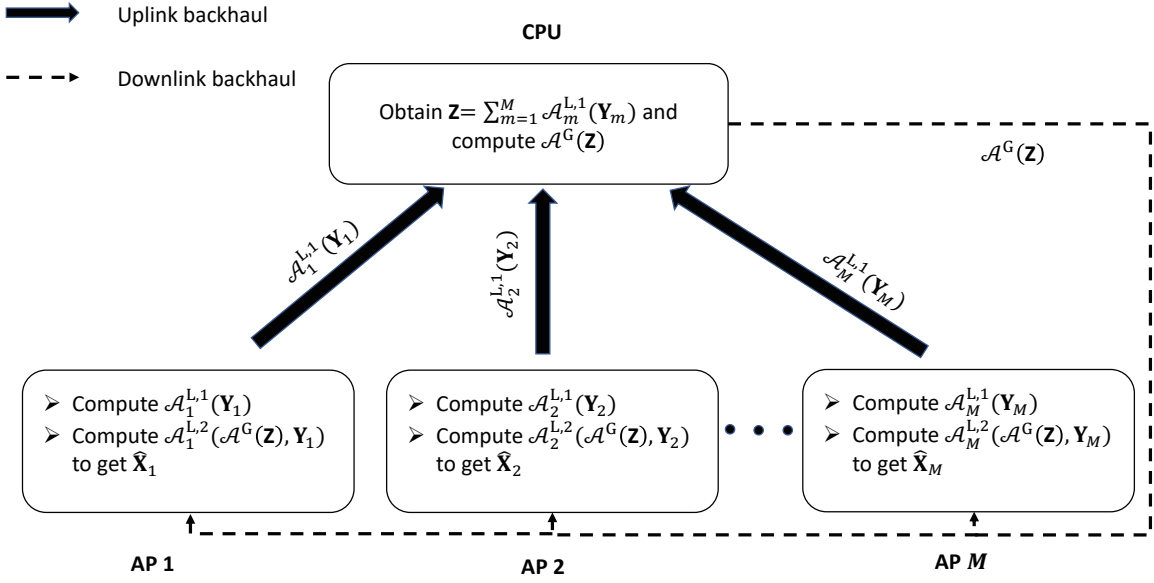}
    \caption{The global-local computation model of the proposed privacy-preserving noisy matrix completion algorithm.}
    \label{fig:compModel}
\end{figure}

With $\widehat{\mathbf{X}}_m$, each AP $m$ can proceed to estimate its channel $\mathbf{H}_m$ as follows. Recall that $\widehat{\mathbf{X}}_m\equiv \left[\widehat{\mathbf{X}}_{m,{\rm{p}}},\widehat{\mathbf{X}}_{m,{\rm{d}}}\right]$, where $\widehat{\mathbf{X}}_{m,{\rm{p}}}\in\mathbb{C}^{{N_{\rm{a}}}\times \tau_{\rm{p}}}$ and
$\widehat{\mathbf{X}}_{m,{\rm{d}}}\in\mathbb{C}^{{N_{\rm{a}}}\times \tau_{\rm{d}}}$ denote the estimates of
$\mathbf{X}_{m,{\rm{p}}}=\mathbf{H}_m\mathbf{P}$ and $\mathbf{X}_{m,{\rm{d}}}=\mathbf{H}_m\mathbf{D}$ respectively. Then the estimate of $\mathbf{H}_m$ is given by
\begin{equation}
\label{eq_ce}
    \widehat{\mathbf{H}}_m=\widehat{\mathbf{X}}_{m,{\rm{p}}}{\mathbf{P}}^{\dagger},\forall m\in\mathcal{M},
\end{equation}
where ${\mathbf{P}}^{\dagger}$ denotes the pseudo-inverse of $\mathbf{P}$ which is pre-stored in each AP. Finally, for data detection, the optimal scheme is to compute $\widehat{\mathbf{D}}=\mathbf{\widehat{H}}^{\dagger}{\widehat{\mathbf{X}}_{\rm{d}}}$ at the CPU. However, since the privacy constraint prevents the CPU from having access to $\{\widehat{\mathbf{H}}_m\}$, we let each AP compute the local statistic
\begin{equation}
\label{eq_dt}
    \widehat{\mathbf{D}}_m={\widehat{\mathbf{H}}_m}^{\dagger }\widehat{\mathbf{X}}_{m,{\rm{d}}},\forall m\in\mathcal{M}.
\end{equation}
Then $\left\{\widehat{\mathbf{D}}_m\right\}$ are sent to the CPU which performs data detection based on the combined statistic
\begin{equation}
\label{equ_dtc}
    \widehat{\mathbf{D}}=\frac{1}{M}\sum\limits_{m=1}^M    \widehat{\mathbf{D}}_m.
\end{equation}

Next we provide two privacy-preserving noisy matrix completion algorithms based on the FW algorithm and SVD algorithm, respectively, both of which are in the form of the global-local computation model.

\subsection{FW-based Privacy-preserving Matrix Completion}
Recall that in the absence of noise, we have $\mathbf{Y} = \left(\mathbf{X}\right)_\Omega$ and ${\rm{rank}}\left(\mathbf{X}\right)\leq K$. 
However, the rank constraint is nonconvex.
A popular approach to noisy matrix completion is based on the following least-squares formulation with the nuclear norm constraint\cite{freund2017extended,jain2017differentially,zheng2018distributed,wai2017decentralized}
\begin{equation}
\label{equ_Pnorm}
\widehat{\mathbf{X}}=\arg\min_{\left\|\mathbf{X}\right\|_{\rm{nuc}}\le K}  \left\|\left(\mathbf{X}\right)_{\Omega}-\mathbf{Y}\right\|_{\mathcal{F}}^2,
\end{equation}
which is a convex problem.
The Frank-Wolfe (FW) algorithm is an iterative procedure to solve (\ref{equ_Pnorm}), given by
\begin{equation}
\label{equ_Yupd}
    {\mathbf{X}}^{(n)}=\left(1-\eta^{(n)}\right)\mathbf{X}^{(n-1)}-\frac{\eta^{(n)}{K}}{\lambda^{(n)}}{\mathbf{J}^{(n-1)}}\mathbf{v}^{(n)}\left(\mathbf{v}^{(n)}\right)^{\rm{H}},
\end{equation}
where $\eta^{(n)}$ is the step size at the $n$-th iteration; $\mathbf{X}^{(0)}=\mathbf{0}_{M{N_{\rm{a}}}\times {\tau_{\rm{c}}}}$; $\mathbf{J}^{(n-1)}=\left(\mathbf{X}^{(n-1)}\right)_{\Omega}-\mathbf{Y}$. $\lambda^{(n)}$ and $\mathbf{v}^{(n)}\in\mathbb{C}^{{\tau_{\rm{c}}}\times 1}$ are respectively the largest singular value and right-singular vector of ${\mathbf{J}^{(n-1)}}$. Hence each FW update is in terms of a rank-one matrix with nuclear norm at most $K$. After $T$ iterations, the rank of ${\mathbf{X}}^{(T)}$ is at most $T$. In addition, it is known that ${\mathbf{X}}^{(T)}$ can approach the optimal solution to (\ref{equ_Pnorm}) when $T$ is large\cite{jain2017differentially}.

Recall that $M{N_a}>\tau_{\rm{c}}$, thus $\left(\lambda^{(n)}\right)^2$ and $\mathbf{v}^{(n)}$ are the largest eigenvalue and eigenvector of $\left({\mathbf{J}^{(n-1)}}\right)^{\rm{H}}{\mathbf{J}^{(n-1)}}=\sum\limits_{m=1}^M\left({\mathbf{J}_m^{(n-1)}}\right)^{\rm{H}}{\mathbf{J}_m^{(n-1)}}$, respectively, where
\begin{equation}
\label{equ_Jm}
    \mathbf{J}_m^{(n-1)}=\left(\mathbf{X}_m^{(n-1)}\right)_{\Omega_m}-\mathbf{Y}_m,
\end{equation}
which can be computed at AP $m$. Hence, (\ref{equ_Yupd}) can be rewritten into a global-local structure as
\begin{equation}
\label{equ_Ymupd}
    {\mathbf{X}}_{m}^{(n)}=\left(1-\eta^{(n)}\right)\mathbf{X}_m^{(n-1)}-\frac{\eta^{(n)}{K}}{\lambda^{(n)}}{\mathbf{J}_m^{(n-1)}}\mathbf{v}^{(n)}\left(\mathbf{v}^{(n)}\right)^{\rm{H}}, \forall m\in\mathcal{M}.
\end{equation}
At the $(n-1)$th iteration, each AP $m$ computes
\begin{equation}
\label{equ_JJm}
    \mathbb{J}_m^{(n-1)}={\left(\mathbf{J}_m^{(n-1)}\right)^{\rm{H}}\mathbf{J}_m^{(n-1)}},
\end{equation}
and sends it to the CPU. Then the CPU computes the largest eigenvalue $\lambda^{(n)}$ and corresponding eigenvector $\mathbf{v}^{(n)}$ of  $\mathbb{W}^{(n-1)}=\sum\limits_{m=1}^M{\mathbb{J}_m^{(n-1)}}$. At the $n$-th iteration, $\lambda^{(n)}$ and $\mathbf{v}^{(n)}$ are broadcast to all APs and then each AP $m$ computes ${\mathbf{X}}_{m}^{(n)}$ according to (\ref{equ_Ymupd}). However, because ${\mathbb{J}_m^{(n-1)}}$ contains the private data $\mathbf{Y}_m$, it should not be sent directly from AP $m$ to the CPU. 

\begin{algorithm}
\caption{FW-based Privacy-preserving Matrix Completion Algorithm}
\KwIn{Privacy parameters $\left(\epsilon,\delta\right)$, number of users $K$, sampled matrix $\mathbf{Y}_m$ in each AP $m$, bound $L$ on $\left\|\mathbf{Y}_m\right\|_{\mathcal{F}}$, number of APs $M$, number of time slots $\tau_{\rm{c}}$, number of iterations $T$}

\KwOut{$\widehat{\mathbf{X}}_m,\forall m\in\mathcal{M}$}

Set
$\mathbf{X}^{(0)}=\mathbf{0}_{M{N_{\rm{a}}}\times {\tau_{\rm{c}}}}$
, $\eta^{(1)}=1$, $\eta^{(n)}=\frac{1}{T}, n=2,\cdots,T$ and $\mu$ given by (\ref{equ_mu})

\For{$n=1:T$}
{
     
    
    $\mathcal{A}_m^{{\rm{L}},1}$: Each AP $m$ computes $\mathbb{J}_m^{(n-1)}$ according to (\ref{equ_JJm}) and (\ref{equ_Jm}) and sends ${\mathbb{\widehat{J}}_m^{(n-1)}}={\mathbb{J}_m^{(n-1)}}+\mathbf{G}_m^{(n-1)}$ to CPU

    $\mathcal{A}^{\rm{G}}$:  CPU receives ${\mathbb{\widehat{J}}_m^{(n-1)}}$ from each AP $m$ and computes $\mathbb{\widehat{W}}^{(n-1)}=\sum\limits_{m=1}^M{\mathbb{\widehat{J}}_m^{(n-1)}}$, then 
    \begin{itemize}
        \item computes the largest eigenvalue $\left(\widehat{\lambda}^{(n)}\right)^2$ and eigenvector $\mathbf{\widehat{v}}^{(n)}$ of $\mathbb{\widehat{W}}^{(n-1)}$\\
        \item computes $\widetilde{\lambda}^{(n)}$ according to (\ref{equ_wlu})
    \end{itemize}
    
    CPU sends $\left(\mathbf{\widehat{v}}^{(n)},\widetilde{\lambda}^{(n)}\right)$ to all APs
    
    $\mathcal{A}_m^{{\rm{L}},2}$: Each AP $m$ computes ${\mathbf{X}}_{m}^{(n)}$ according to (\ref{equ_XmFWPP})
     
}
\label{algo}
\end{algorithm}

To preserve the privacy, we add Gaussian noise with calibrated variance to perturb ${\mathbb{J}_m^{(n-1)}}$. Each AP $m$ sends the perturbed matrix ${\mathbb{\widehat{J}}_m^{(n-1)}}={\mathbb{J}_m^{(n-1)}}+\mathbf{G}_m^{(n-1)}$ to the CPU, where $\mathbf{G}_m^{(n-1)}$ is a $\tau_{\rm{c}}\times \tau_{\rm{c}}$ Hermitian noise matrix, whose upper triangular and diagonal elements are respectively i.i.d.  $\mathcal{N}_{\rm{c}}\left(0,\mu^2\right)$ and $\mathcal{N}\left(0,\mu^2\right)$ samples, and the lower triangular elements are complex conjugates of the upper triangular counterparts. To calibrate the perturbation noise variance $\mu$, we need an upper bound on  $\left\|\mathbf{Y}_m\right\|_{\mathcal{F}},\forall m\in\mathcal{M}$, which can be obtained as follows. Note that $\left\|\mathbf{Y}_m\right\|_{\mathcal{F}}\le \left\|\mathbf{R}_m\right\|_{\mathcal{F}}\le\left\|\mathbf{H}_m{\mathbf{S}}\right\|_{\mathcal{F}}+\left\|\mathbf{N}_m\right\|_{\mathcal{F}}\le \left\|\mathbf{H}_m\right\|_{\mathcal{F}}\left\|{\mathbf{S}}\right\|_{\mathcal{F}}+\left\|\mathbf{N}_m\right\|_{\mathcal{F}}$.
In the massive MIMO regime, due to channel hardening, we have 
\begin{equation}
\label{equ_Hmap}
    {\left\|\mathbf{H}_m\right\|_{\mathcal{F}}^2}\rightarrow{\mathbb{E}\left\{\left\|\mathbf{H}_m\right\|_{\mathcal{F}}^2\right\}}={N_a}\sum\nolimits_{k=1}^K{\beta_{k,m}}.
\end{equation}
Moreover, since each data symbol in $\mathbf{S}$ has unit power and each noise sample in $\mathbf{N}_m$ has variance 
$\sigma^2$, by the law of large numbers, we have
\begin{subequations}
\begin{align}
    {\left\|{\mathbf{S}}\right\|_{\mathcal{F}}^2}&\rightarrow{K{\tau_{\rm{c}}}},\label{equ_Sap}\\
    {\left\|{\mathbf{N}_m}\right\|_{\mathcal{F}}^2}&\rightarrow{{{N_a}{\tau_{\rm{c}}}}\sigma^2}.
\end{align}
\end{subequations}
Hence, when $K{\tau_{\rm{c}}}$ and ${N_a}{\tau_{\rm{c}}}$ are sufficiently large, we can bound $\left\|\mathbf{Y}_m\right\|_{\mathcal{F}},\forall m\in\mathcal{M}$ by
\begin{equation}
\label{equ_L}
    L=\max_{m\in\mathcal{M}}\sqrt{K{\tau_{\rm{c}}}{N_a}\sum\nolimits_{k=1}^K{\beta_{k,m}}}+ \sqrt{{N_a}{\tau_{\rm{c}}}\sigma^2}.
\end{equation}
Note that each AP $m$ sends a perturbed matrix which contains the private signal $\mathbf{Y}_m$ to the CPU, for a total of $T$ iterations. Hence, the information released to the CPU can be regarded as a $T$-fold composition. Then by Lemma \ref{lem_TFC} and \ref{lem_GM}, in order to achieve $\left(\epsilon,\delta\right)$-joint DP, the perturbation noise variance is calibrated as
\begin{equation}
\label{equ_mu}
    \mu= {16L^2\sqrt{\frac{T}{M}\ln{\left(\frac{2.5T}{\delta}\right)}\ln{\left(\frac{2}{\delta}\right)}}/{\epsilon}}.
\end{equation}

Next the CPU computes the largest eigenvalue $\left(\widehat\lambda^{(n)}\right)^2$ and eigenvector $\mathbf{\widehat{v}}^{(n)}$  of  $\mathbb{\widehat{W}}^{(n-1)}=\sum\limits_{m=1}^M{\mathbb{\widehat{J}}_m^{(n-1)}}$. To control the error introduced by the perturbation noise, $\widehat\lambda^{(n)}$ is replaced by\cite{jain2017differentially}
\begin{equation}
\label{equ_wlu}
    \widetilde{\lambda}^{(n)}=\widehat\lambda^{(n)}+ \sqrt{\mu} \left({M\tau_{\rm{c}}}\right)^{1/4}.
\end{equation} 
Finally, the privacy-preserving implementation of (\ref{equ_Ymupd}) can be written as
\begin{equation}
\label{equ_XmFWPP}
    {\mathbf{X}}_{m}^{(n)}=\Xi_{L,\Omega_m}\left(\left(1-\eta^{(n)}\right)\mathbf{X}_m^{(n-1)}-\frac{ \eta^{(n)}{K}}{\widetilde{\lambda}^{(n)}}{\mathbf{J}_m^{(n-1)}\mathbf{\widehat{v}}^{(n)}}\left(\mathbf{\widehat{v}}^{(n)}\right)^{\rm{H}}\right),\forall m\in\mathcal{M},
\end{equation}
where the operator $\Xi_{L,\Omega_m}\left({\mathbf{M}}\right)=\min\left\{\frac{L}{\left\|\left({\mathbf{M}}\right)_{\Omega_m}\right\|_{\mathcal{F}}},1\right\}{{\mathbf{M}}}$ ensures $\left\|\left(\mathbf{X}_m^{(n)}\right)_{\Omega_m}\right\|_{\mathcal{F}}\le L$.

The proposed FW-based privacy-preserving matrix completion  algorithm is summarized in Algorithm 1.


\subsection{SVD-based Privacy-preserving Matrix Completion Algorithm}
Here we consider another approach to low-rank matrix completion based on singular value decomposition (SVD) that yields a solution with bounded error \cite{keshavan2010matrix}. It first trims $\mathbf{Y}$ to obtain $\widetilde{\mathbf{Y}}$ by setting to zero all rows with more than ${2{N_r}{\tau_{\rm{c}}}}/{N_a}$ non-zero entries, where recall that $N_r$ is the number of RF chains at each AP, i.e., the number of non-zero elements in each column of $\mathbf{Y}$. Then it performs SVD on $\widetilde{\mathbf{Y}}$ and denotes $\mathbf{{V}}_K\in\mathbb{C}^{\tau_{\rm{c}}\times K}$ as the matrix consisting of the right singular vectors corresponding to the $K$ largest singular values. 
Finally, the completed matrix is given by 
\begin{equation}
\label{equ_Xsvd}
    \widehat{\mathbf{X}}=\frac{N_a}{N_r}\widetilde{\mathbf{Y}}\mathbf{V}_K\mathbf{V}_K^{\rm{H}}.
\end{equation}
It is clear that ${\rm{rank}}\left(\widehat{\mathbf{X}}\right)=K$.

Note that $\mathbf{V}_K$ can also be obtained as the eigenvectors corresponding to the $K$ largest eigenvalues of $\widetilde{\mathbf{Y}}^{\rm{H}}\widetilde{\mathbf{Y}}=\sum\limits_{m=1}^M\widetilde{\mathbf{Y}}_m^{\rm{H}}\widetilde{\mathbf{Y}}_m$, where $\widetilde{\mathbf{Y}}_m^{\rm{H}}\widetilde{\mathbf{Y}}_m$ is computed at AP $m$ and then sent to the CPU. However, because ${\mathbb{J}}_{m}=\widetilde{\mathbf{Y}}_m^{\rm{H}}\widetilde{\mathbf{Y}}_m$ contains the private data $\mathbf{Y}_m$, it should not be sent directly from AP $m$ to the CPU. Similarly, we also add Gaussian noise to perturb it. Hence, the released data from each AP $m$ to the CPU is ${\mathbb{\widehat{J}}_m}={\mathbb{J}_m}+\mathbf{G}_m$, where
$\mathbf{G}_m$ is a $\tau_{\rm{c}}\times \tau_{\rm{c}}$ Hermitian noise matrix, whose upper triangular and diagonal elements are respectively i.i.d.  $\mathcal{N}_{\rm{c}}\left(0,\nu^2\right)$ and $\mathcal{N}\left(0,\nu^2\right)$ samples, and the lower triangular elements are complex conjugates of the upper triangular counterparts. Hence, the variance of total perturbation noise at the CPU is $M\nu^2$. By Lemma \ref{lem_GM}, in order to achieve $\left(\epsilon,\delta\right)$-joint DP, $\nu$ is calibrated as
\begin{equation}
\label{equ_nu}
    \nu= {L^2\sqrt{\frac{2}{M}\ln{\left(\frac{1.25}{\delta}\right)}}/{\epsilon}}.
\end{equation}
Then the CPU computes the $K$ largest eigenvectors of $\mathbf{\widehat{V}}_K$ of the matrix $\mathbb{\widehat{W}}=\sum\limits_{m=1}^M{\mathbb{\widehat{J}}_m}$ and broadcasts it to all APs. Hence, (\ref{equ_Xsvd}) can be implemented in a privacy-preserving form as
\begin{equation}
\label{equ_XmSVD}
    \widehat{\mathbf{X}}_m=\frac{{N_a}}{N_r}\widetilde{\mathbf{Y}}_m{\mathbf{\widehat{V}}_K}{\mathbf{\widehat{V}}_K^{\rm{H}}}, \quad\forall m\in\mathcal{M}.
\end{equation}

The proposed SVD-based privacy-preserving matrix completion algorithm is summarized in Algorithm 2.

\begin{algorithm}[htb]
\caption{SVD-based Privacy-preserving Matrix Completion Algorithm}
\KwIn{Privacy parameters $\left(\epsilon,\delta\right)$, number of users $K$, sampled matrix $\mathbf{Y}_m$ in each AP $m$, bound $L$ on $\left\|\mathbf{Y}_m\right\|_{\mathcal{F}}$, number of APs $M$, number of time slots $\tau_{\rm{c}}$}

\KwOut{$\widehat{\mathbf{X}}_m,\forall m\in\mathcal{M}$}

Set $\nu$ given by (\ref{equ_nu}) 

$\mathcal{A}_m^{{\rm{L}},1}$: AP $m$ trims $\mathbf{Y}_m$ to obtain $\widetilde{\mathbf{Y}}_m$, then computes ${\mathbb{J}_m}=\widetilde{\mathbf{Y}}_m^{\rm{H}}\widetilde{\mathbf{Y}}_m$ and sends ${\mathbb{\widehat{J}}_m}={\mathbb{J}_m}+\mathbf{G}_m$ to the CPU

$\mathcal{A}^{{\rm{G}}}$:  CPU receives ${\mathbb{\widehat{J}}_m}$ from each AP $m$ and computes the $K$ largest eigenvectors $\mathbf{\widehat{V}}_K$ of $\mathbb{\widehat{W}}=\sum\limits_{m=1}^M{\mathbb{\widehat{J}}_m}$. Then $\mathbf{\widehat{V}}_K$ is broadcast to all APs\

$\mathcal{A}_m^{{\rm{L}},2}$: AP $m$ computes
$\widehat{\mathbf{X}}_m$ according to (\ref{equ_XmSVD})\
\label{algo}
\end{algorithm}

\subsection{Computational Complexity and Communication Overhead}
For both Algorithm 1 and Algorithm 2, the computations $\mathcal{A}_m^{{\rm{L}},1}$ and $\mathcal{A}_m^{{\rm{L}},2}$ at each AP $m$ involve only additions and multiplications. The computation $\mathcal{A}^{{\rm{G}}}$ at the CPU in Algorithm 1 involves computing the largest eigen component of an $\tau_{\rm{c}}\times \tau_{\rm{c}}$ matrix in each iteration, for a total of $T$ iterations; whereas in Algorithm 2, $\mathcal{A}^{{\rm{G}}}$ involves computing the $K$ largest eigen components of an $\tau_{\rm{c}}\times \tau_{\rm{c}}$ matrix only once. 

We next compare the communication overheads of the two algorithms. For Algorithm 1, in each iteration, each AP sends the ${\tau_{\rm{c}}\times \tau_{\rm{c}}}$ matrix ${\mathbb{\widehat{J}}_m^{(n)}}$ to the CPU; and the CPU broadcasts the scalar $\widetilde\lambda^{(n)}$ and the ${\tau_{\rm{c}}\times 1}$ vector $\mathbf{\widehat{v}}^{(n)}$ to all APs, for a total of $T$ iterations. For Algorithm 2, each AP sends the ${\tau_{\rm{c}}\times \tau_{\rm{c}}}$ matrix ${\mathbb{\widehat{J}}_m}$ to the CPU only once; and the CPU broadcasts the ${\tau_{\rm{c}}\times K}$ matrix $\mathbf{\widehat{V}}_K$ to all APs only once. Hence the ratio of the broadcast overhead from the CPU to all APs between Algorithms 1 and 2 is $T:K$; and the ratio of the unicast overhead from each AP to the CPU is $T:1$.

\section{Privacy and Estimation Error Tradeoff Analysis}
In this section, we first show that both Algorithms 1 and 2 are $\left(\epsilon,\delta\right)$-joint differentially private. We then provide their channel estimation error bounds in terms of the privacy parameters.
\subsection{Privacy Analysis}
\begin{thm}
Algorithm 1 is $\left(\epsilon,\delta\right)$-joint differentially private. 
\end{thm}
\begin{proof}
First, we prove that in each iteration, the information released to the CPU is differentially private. Specifically, the received signal by the CPU at the $n$-th iteration is
\begin{equation}
    \sum\nolimits_{m=1}^M{\mathbb{\widehat{J}}_m^{(n-1)}}=\sum\nolimits_{m=1}^M{\mathbb{{E}}_m^{(n-1)}}+\sum\nolimits_{m=1}^M\mathbf{G}_m^{(n-1)},
\end{equation}
where $\sum\nolimits_{m=1}^M\mathbf{G}_m^{(n-1)}$ is the total perturbation noise from the APs, which is a Hermitian matrix. Its  upper triangular and diagonal elements are respectively i.i.d.  $\mathcal{N}_{\rm{c}}\left(0,M\mu^2\right)$ and $\mathcal{N}\left(0,M\mu^2\right)$ samples, and the lower triangular elements are complex conjugates of the upper triangular counterparts, where
\begin{equation}
    M{\mu}^2=\frac{\left(4L^2\right)^2{2\ln{\left(1.25/\frac{\delta}{2T}\right)}}}{\left({\epsilon}/{\sqrt{8T\ln{\left(2/\delta\right)}}}\right)^2}.
\end{equation}

Recall that $\sum\nolimits_{m=1}^M{\mathbb{{J}}_m^{(n-1)}}=\sum\nolimits_{m=1}^M{\left(\mathbf{J}_m^{(n-1)}\right)^{\rm{H}}\mathbf{J}_m^{(n-1)}}$ with $\mathbf{J}_m^{(n-1)}=\left(\mathbf{X}_m^{(n-1)}\right)_{\Omega_m}-\mathbf{Y}_m$; and for each AP $m\in\mathcal{M}$, we have $\left\|\mathbf{Y}_m\right\|_{\mathcal{F}}\le L$ and $\left\|\left(\mathbf{X}_m^{(n-1)}\right)_{\Omega_m}\right\|_{\mathcal{F}}\le L$. Hence, the ${\ell _2}$-sensitivity of the signal $\sum\nolimits_{m=1}^M{\mathbb{{J}}_m^{(n-1)}}$ is $4L^2$. Then according to Lemma \ref{lem_GM}, the information released to the CPU at the $n$-th iteration, $\sum\nolimits_{m=1}^M{\widehat{\mathbb{{J}}}_m^{(n)}}$, is $\left({{\epsilon}/{\sqrt{8T\ln{\left(2/\delta\right)}}}},\delta/2T\right)$-differentially private.
The CPU computes the dominant eigen components of the differentially private $\sum\nolimits_{m=1}^M{\widehat{\mathbb{{J}}}_m^{(n)}}$ and by Lemma \ref{lem_PP}, the obtained $\widetilde\lambda^{(n)}$ and $\mathbf{\widehat{v}}^{(n)}$ are also $\left({{\epsilon}/{\sqrt{8T\ln{\left(2/\delta\right)}}}},\delta/2T\right)$-differentially private. From Line 5 in Algorithm 1, each AP $m$ computes $\mathbf{X}_m^{(n)}$ using differentially private $\widetilde\lambda^{(n)}$ and $\mathbf{\widehat{v}}^{(n)}$ and its local signal as input. By Lemma 3, $\mathbf{X}_m^{(n)}$ is then $\left({{\epsilon}/{\sqrt{8T\ln{\left(2/\delta\right)}}}},\delta/2T\right)$-joint differentially private. 

Finally, because $\widehat{\mathbf{X}}_m=\mathbf{X}_m^{(T)}$ is the result of a $T$-fold composition of $\left({{\epsilon}/{\sqrt{8T\ln{\left(2/\delta\right)}}}},\delta/2T\right)$-joint differentially private algorithms, by Lemma \ref{lem_TFC}, $\widehat{\mathbf{X}}_m,\forall m\in\mathcal{M}$ satisfy $\left(\epsilon,\delta\right)$-joint DP. $\hfill\Box$
\end{proof}

Similar privacy analysis can be carried out for Algorithm 2 and we arrive at the following result.
\begin{thm}
Algorithm 2 is $\left(\epsilon,\delta\right)$-joint differentially private. 
\end{thm}

Hence in both matrix completion algorithms, by adding Gaussian noise with calibrated variance to perturb the released information from the APs to the CPU, joint differential privacy can be achieved. However, the perturbation noise inevitably increases the matrix completion error and therefore the channel estimation error. Next, we will provide an analysis on the estimation error bounds for Algorithms 1 and 2. In particular, for both algorithms, we will show that the estimation error decreases with the increase of the data payload size $\tau_{\rm{d}}$. 

\subsection{Error Bounds and Scaling Laws for Channel Estimation}
The error bounds for the FW and SVD-based noisy matrix completions are provided in \cite{freund2017extended} and \cite{keshavan2010matrix} respectively, both of which do not consider privacy. \cite{jain2017differentially} gives the error bounds for both FW and SVD-based differentiately private matrix completions, but the matrix is noise-free. Here we analyze the error bounds for differentiately private noisy matrix completions.

\subsubsection{Error Bound and Scaling Law of Algorithm 1}
The key step of Algorithm 1 in (\ref{equ_XmFWPP}) can be rewritten in terms of the whole matrix as
\begin{equation}
\label{equ_XupSVDPPt}
    {\mathbf{X}}^{(n)}=\Xi_{L,\Omega}\left(\left(1-\eta^{(n)}\right)\mathbf{X}^{(n-1)}-\frac{ \eta^{(n)}{K}}{\widetilde{\lambda}^{(n)}}{\mathbf{J}^{(n-1)}\mathbf{\widehat{v}}^{(n)}}\left(\mathbf{\widehat{v}}^{(n)}\right)^{\rm{H}}\right).
\end{equation}
Following Lemma D.5 in \cite{jain2017differentially}, the projection operator $\Xi_{L,\Omega}$ does not introduce additional error. Hence, we will ignore it in the following analysis. Comparing (\ref{equ_XupSVDPPt}) and  (\ref{equ_Yupd}), we can see that the privacy-preserving FW algorithm replaces $\mathbf{O}^{(n)}=-\frac{K}{\lambda^{(n)}}{\mathbf{J}^{(n-1)}\mathbf{v}^{(n)}}\left(\mathbf{v}^{(n)}\right)^{\rm{H}}$ in the original FW algorithm with $\mathbf{Q}^{(n)}=-\frac{ K}{\widetilde\lambda^{(n)}}{\mathbf{J}^{(n-1)}\mathbf{\widehat{v}}^{(n)}}\left(\mathbf{\widehat{v}}^{(n)}\right)^{\rm{H}}$. To quantify the additional error caused by this replacement, we first provide the following two lemmas, which are generalizations of Lemma D.4 in \cite{jain2017differentially} and Theorem 1 in \cite{jaggi2013revisiting} to the case of noisy matrix completion.

\begin{lem}
\label{lem_diffe}
The following bound holds with high probability
\footnote{``with high probability'' means with probability $1-1/{\tau_{\rm{c}}}^{\Theta\left(1\right)}$.}
\begin{equation}
    \left<\mathbf{Q}^{(n)},\frac{1}{\left|\Omega\right|}\mathbf{J}^{(n-1)}\right>_{\mathcal{F}}-\left<\mathbf{O}^{(n)},\frac{1}{\left|\Omega\right|}\mathbf{J}^{(n-1)}\right>_{\mathcal{F}}\le O\left(\frac{K}{\left|\Omega\right|}\sqrt{\mu}\left({M}{\tau_{\rm{c}}}\right)^{1/4}\right),\forall n.
\end{equation}
\end{lem}

\begin{proof}
See Appendix A.
\end{proof}

\begin{lem}
\label{lem_ugFW}
If the update of $\mathbf{X}^{(n)}$ in (\ref{equ_Yupd}) is modified as
 \begin{equation}
 \label{equ_XupW}
     {\mathbf{X}}^{(n)}=\left(1-\eta^{(n)}\right)\mathbf{X}^{(n-1)}+\eta^{(n)}{\mathbf{W}}^{(n)},
 \end{equation}
where ${\mathbf{W}}^{(n)}$ satisfies $\left\|{\mathbf{W}}^{(n)}\right\|_{\rm{nuc}}\le K, \forall n$ and 
\begin{equation}
\label{equ_gammaC}
    \left<\mathbf{W}^{(n)},\frac{1}{\left|\Omega\right|}\mathbf{J}^{(n-1)}\right>_{\mathcal{F}}- \left<\mathbf{O}^{(n)},\frac{1}{\left|\Omega\right|}\mathbf{J}^{(n-1)}\right>_{\mathcal{F}}\le \gamma, \forall n,
\end{equation}
where $\left<\mathbf{A},\mathbf{B}\right>_{\mathcal{F}}={\rm{tr}}\left(\mathbf{A}^{\rm{H}}\mathbf{B}\right)$ is the Frobenius inner product. Then for $n=T$, we have
\begin{equation}
\label{equ_xkee}
\frac{1}{\left|\Omega\right|}\left\|\left(\mathbf{X}^{(T)}\right)_{\Omega}-\mathbf{Y}\right\|_{\mathcal{F}}^2\le  4\gamma+\frac{K^2}{\left|\Omega\right|}+ \frac{K^2}{\left|\Omega\right|T}+{\sigma^2}.
\end{equation}
\end{lem}
\begin{proof}
See Appendix B.
\end{proof}

Using Lemmas \ref{lem_ugFW} and \ref{lem_diffe},  we can then obtain the error bound for Algorithm 1 as follows. 
\begin{thm}
\label{thm_OE}
Denote $\widehat{\mathbf{X}}=\mathbf{X}^{(T)}$ as the output of Algorithm 1. Then the following error bound on the observed entries in $\Omega$ holds with high probability
\begin{equation}
\label{equ_ebo}
      \frac{1}{{\left|\Omega\right|}} \left\|\left(\widehat{\mathbf{X}}\right)_{\Omega}-\left(\mathbf{X}\right)_{\Omega}\right\|_{\mathcal{F}}^2=
        O\left({\frac{4K}{\left|\Omega\right|}\sqrt{\mu}\left({M}{\tau_{\rm{c}}}\right)^{1/4}+\frac{K^2}{\left|\Omega\right|}+\frac{K^2}{{\left|\Omega\right|}T}+2{\sigma^2}}\right).
\end{equation}

Furthermore, the generalization error $\mathcal{E}\left(\widehat{\mathbf{X}}\right)=\frac{1}{M{N_a}\tau_{\rm{c}}}\left\|\widehat{\mathbf{X}}-\mathbf{X}\right\|_{\mathcal{F}}^2$
is bounded as follows with high probability 
\begin{equation}
\label{equ_XeFW}
    \mathcal{E}\left(\widehat{\mathbf{X}}\right)=\widetilde{O}\left({\frac{4K}{\left|\Omega\right|}\sqrt{\mu}\left({M}{\tau_{\rm{c}}}\right)^{1/4}+\frac{K^2}{\left|\Omega\right|}+\frac{K^2}{{\left|\Omega\right|}T}+2{\sigma^2}}+\sqrt{\frac{T\left(M{N_a}+\tau_{\rm{c}}\right)}{\left|\Omega\right|}}\right),
\end{equation}
where $\widetilde{O}\left(\cdot\right)$ hides poly-logarithmic terms in $M{N_a}$ and $\tau_{\rm{c}}$. When the number of iterations is chosen as
$T=O\left(\frac{K^{4/3}}{\left(\left|\Omega\right|\left(M{N_a}+{\tau_{\rm{c}}}\right)\right)^{1/3}}\right)$, we can obtain the  generalization error bound as
\begin{equation}
\label{equ_XeFWb}
    \mathcal{E}\left(\widehat{\mathbf{X}}\right)=\widetilde{O}\left({\frac{4K}{\left|\Omega\right|}\sqrt{\mu}\left({M}{\tau_{\rm{c}}}\right)^{1/4}+\frac{K^2}{\left|\Omega\right|}+2{\sigma^2}}+{\frac{2K^{2/3}\left(M{N_a}+\tau_{\rm{c}}\right)^{1/3}}{\left|\Omega\right|^{2/3}}}\right).
\end{equation}
\end{thm}
\begin{proof}
Since $\mathbf{Q}^{(n)}=-\frac{ K}{\widetilde\lambda^{(n)}}{\mathbf{J}^{(n-1)}\mathbf{\widehat{v}}^{(n)}}\left(\mathbf{\widehat{v}}^{(n)}\right)^{\rm{H}}$ is rank-one, we have 
$\left\|\mathbf{Q}^{(n)}\right\|_{\rm{nuc}}=\left\|\mathbf{Q}^{(n)}\right\|_{\mathcal{F}}$, and
\begin{equation}
\begin{aligned}
    \left\|\mathbf{Q}^{(n)}\right\|_{\mathcal{F}}^2&={\rm{tr}}\left(\left(\mathbf{Q}^{(n)}\right)^{\rm{H}}\mathbf{Q}^{(n)}\right)\\
    &=\frac{K^2}{\left(\widetilde{\lambda}^{(n)}\right)^2}{\rm{tr}}\left({\mathbf{\widehat{v}}^{(n)}}\left(\mathbf{\widehat{v}}^{(n)}\right)^{\rm{H}}\left({\mathbf{J}^{(n-1)}}\right)^{\rm{H}}{\mathbf{J}^{(n-1)}\mathbf{\widehat{v}}^{(n)}}\left(\mathbf{\widehat{v}}^{(n)}\right)^{\rm{H}}\right)\\
    &={K^2}\frac{\left(\mathbf{\widehat{v}}^{(n)}\right)^{\rm{H}}\left({\mathbf{J}^{(n-1)}}\right)^{\rm{H}}{\mathbf{J}^{(n-1)}\mathbf{\widehat{v}}^{(n)}}}{\left(\widetilde{\lambda}^{(n)}\right)^2}\left(\mathbf{\widehat{v}}^{(n)}\right)^{\rm{H}}\mathbf{\widehat{v}}^{(n)}\\
    &={K^2}\frac{\left\|{\mathbf{J}^{(n-1)}\mathbf{\widehat{v}}^{(n)}}\right\|_{\mathcal{F}}^2}{\left(\widetilde{\lambda}^{(n)}\right)^2}.
\end{aligned}
\end{equation}
According to Lemma D.2 in \cite{jain2017differentially}, the following holds with high probability
\begin{equation}
    \left\|\mathbf{J}^{(n-1)}\mathbf{\widehat{v}}^{(n)}\right\|_{\mathcal{F}}\le \widehat{\lambda}+O\left(\sqrt{\mu}\left(M{\tau_{\rm{c}}}\right)^{1/4}\right).
\end{equation}
Hence, according to the definition of $\widetilde{\lambda}^{(n)}$ in (\ref{equ_wlu}), with high probability, we have $\left\|\mathbf{Q}^{(n)}\right\|_{\rm{nuc}}\le K, \forall n$. Then by Lemma \ref{lem_diffe} and Lemma \ref{lem_ugFW}, the following holds with high probability
\begin{equation}
    \frac{1}{\left|\Omega\right|}\left\|\left(\widehat{\mathbf{X}}\right)_{\Omega}-\mathbf{Y}\right\|_{\mathcal{F}}^2=
O\left({\frac{4K}{\left|\Omega\right|}\sqrt{\mu}\left({M}{\tau_{\rm{c}}}\right)^{1/4}+\frac{K^2}{\left|\Omega\right|}+\frac{K^2}{{\left|\Omega\right|}T}+{\sigma^2}}\right).
\end{equation}

Due to the triangular inequality property, we have
\begin{equation}
\label{equ_Nbound}
\begin{aligned}
    \left\|\left(\widehat{\mathbf{X}}\right)_{\Omega}-\left(\mathbf{X}\right)_{\Omega}\right\|_{\mathcal{F}}^2&\le\left(\left\|\left(\widehat{\mathbf{X}}\right)_{\Omega}-\mathbf{Y}\right\|_{\mathcal{F}}+\left\|\mathbf{Y}-\left(\mathbf{X}\right)_{\Omega}\right\|_{\mathcal{F}}\right)^2\\
    &\le 2\left\|\left(\widehat{\mathbf{X}}\right)_{\Omega}-\mathbf{Y}\right\|_{\mathcal{F}}^2+2\left\|\mathbf{Y}-\left(\mathbf{X}\right)_{\Omega}\right\|_{\mathcal{F}}^2.
\end{aligned}
\end{equation}
Note that $\left\|\mathbf{Y}-\left(\mathbf{X}\right)_{\Omega}\right\|_{\mathcal{F}}^2=\left\|\left(\mathbf{N}\right)_{\Omega}\right\|_{\mathcal{F}}^2\rightarrow {\left|\Omega\right|\sigma^2}$ when $\left|\Omega\right|=M{N_r}{\tau_{\rm{c}}}$ is large. Hence,
with high probability (\ref{equ_ebo}) holds.

Note that (\ref{equ_ebo}) gives the error bound on observed entries in $\Omega$. Using Theorem 1 in \cite{srebro2005rank}, we can then generalize the error bound to the entire matrix $\widehat{\mathbf{X}}$ given by (\ref{equ_XeFW}). It can be seen that the third term in (\ref{equ_XeFW}) decreases with $T$, while the last term increases with $T$. By setting these two terms as the same order, we obtain  $T=O\left(\frac{K^{4/3}}{\left(\left|\Omega\right|\left(M{N_a}+{\tau_{\rm{c}}}\right)\right)^{1/3}}\right)$, and the corresponding generalization error bound in (\ref{equ_XeFWb}).
$\hfill\Box$
\end{proof}

\begin{rem}
To achieve stronger privacy level, i.e., smaller $\epsilon$ and/or $\delta$, the perturbation noise variance $\mu$ in (\ref{equ_mu}) will increase, which in turn increases the matrix completion error according to (\ref{equ_XeFW}).
\end{rem}

Note that the matrix completion error in (\ref{equ_XeFWb}) has two sources: the first term is due to the perturbation noise added to achieve DP, and the other terms are the error inherent to the FW algorithm. Both error sources contribute to the channel estimation error through (\ref{eq_ce}). The following result shows that the channel estimation errors caused by both sources decreases with the increase of the payload size $\tau_d$, at different speed, for fixed pilot size $\tau_{\rm{p}}$.
 
\begin{coro}
For fixed privacy parameters $\epsilon$ and $\delta$, and fixed pilot length $\tau_{\rm{p}}$, the estimation error of the proposed privacy-preserving channel estimator that employs Algorithm 1 scales with the data payload size $\tau_{\rm{d}}$ as
\begin{equation}
\label{equ_ceeO}
    \left\|\widehat{\mathbf{H}}_m-\mathbf{H}_m\right\|_{\mathcal{F}}^2=O\left(\tau_{\rm{d}}^{-1/3}\right).
\end{equation}
Moreover, the portion of the estimation error due to the perturbation noise added to preserve privacy scales as $O\left({\tau_{\rm{d}}^{-5/12}}\right)$.
\end{coro}
\begin{proof}
Note that from (\ref{eq_ce}), for a given pilot matrix $\mathbf{P}$, the channel estimation error at AP $m$ satisfies
\begin{equation}
\label{equ_CEE}
\begin{aligned}
    \left\|\widehat{\mathbf{H}}_m-\mathbf{H}_m\right\|_{\mathcal{F}}^2&=\left\|\left(\widehat{\mathbf{X}}_{m,{\rm{p}}}-\mathbf{X}_{m,{\rm{p}}}\right){\mathbf{P}}^{\dagger}\right\|_{\mathcal{F}}^2\\
    &\le \left\|\widehat{\mathbf{X}}_{m,{\rm{p}}}-\mathbf{X}_{m,{\rm{p}}}\right\|_{\mathcal{F}}^2\left\|{\mathbf{P}}^{\dagger}\right\|_{\mathcal{F}}^2.
\end{aligned}
\end{equation}

Now assuming that the matrix completion error is uniform among the entries of $\mathbf{X}$, then we have
\begin{equation}
\label{equ_mcem}
    \left\|\widehat{\mathbf{X}}_{m,{\rm{p}}}-\mathbf{X}_{m,{\rm{p}}}\right\|_{\mathcal{F}}^2={\frac{1}{M}\frac{\tau_{\rm{p}}}{\tau_{\rm{p}}+\tau_{\rm{d}}}}    \left\|\widehat{\mathbf{X}}-\mathbf{X}\right\|_{\mathcal{F}}^2.
\end{equation}

For fixed privacy parameters $\epsilon$ and $\delta$, and fixed pilot length $\tau_{\rm{p}}$, when $T=O\left(\frac{K^{4/3}}{\left(\left|\Omega\right|\left(M{N_a}+{\tau_{\rm{c}}}\right)\right)^{1/3}}\right)$, we have $L=O\left(\sqrt{\tau_{\rm{d}}}\right)$ and $\mu=O\left(\tau_{\rm{d}}^{2/3}\right)$ hiding all other parameters and the logarithmic term  $\sqrt{\ln\left(\tau_{\rm{d}}^{-2/3}\right)}$ according to (\ref{equ_L}) and (\ref{equ_mu}), respectively.
The term in (\ref{equ_XeFWb}) that has $\mu$ is due to the perturbation noise, and scales as $O\left(\tau_{\rm{d}}^{-5/12}\right)$ since $\left|\Omega\right|=M{N_r}{\tau_{\rm{c}}}$. In addition, the last term in (\ref{equ_XeFWb}) scales as $O\left(\tau_{\rm{d}}^{-1/3}\right)$, which dominates the matrix completion error. 
Then by (\ref{equ_CEE}) and (\ref{equ_mcem}), the statements of the corollary hold.  $\hfill\Box$
\end{proof}

\subsubsection{Error Bound and Scaling Law of Algorithm 2}
The key step of Algorithm 2 in (\ref{equ_XmSVD}) can be rewritten in terms of the whole matrix as
\begin{equation}
\label{equ_XSVDp}
    \widehat{\mathbf{X}}=\frac{{N_a}}{N_r}\widetilde{\mathbf{Y}}{\mathbf{\widehat{V}}_K}{\mathbf{\widehat{V}}_K^{\rm{H}}}.
\end{equation}
Compared to the original SVD-based matrix completion in (\ref{equ_Xsvd}), (\ref{equ_XSVDp}) uses $\bm{\widehat\Pi}_K= \mathbf{\widehat{V}}_K\mathbf{\widehat{V}}_K^{\rm{H}}$ instead of $\bm\Pi_K= \mathbf{V}_K\mathbf{V}_K^{\rm{H}}$. Denote the $K$-th and $(K+1)$-th singular values of $\widetilde{\mathbf{Y}}$ as $\lambda_K$ and $\lambda_{K+1}$, respectively. When there is a large
gap between $\lambda_K^2$ and $\lambda_{K+1}^2$, the space spanned by the $K$ largest
eigenvectors of the noise-perturbed version of $\widetilde{\mathbf{Y}}^{\rm{H}}\widetilde{\mathbf{Y}}$, i.e., $\widehat{\mathbb{W}}$ is very close to the space spanned by the $K$ largest right singular vectors of $\widetilde{\mathbf{Y}}$\cite{dwork2014analyze}. In massive MIMO with sufficiently large $M{N_a}$, such a large gap holds and then we have the following matrix completion error bound for Algorithm 2.

\begin{thm}
 If $\lambda_K^2-\lambda_{K+1}^2=\omega\left(\nu\sqrt{M{\tau_{\rm{c}}}}\right)$, then the following error bound on the output $\widehat{\mathbf{X}}$ of Algorithm 2 holds with high probability
\begin{equation}
\label{equ_XerSVD}
        \frac{1}{\sqrt{M{N_a}{\tau_{\rm{c}}}}}\left\|\widehat{\mathbf{X}}-\mathbf{X}\right\|_{\mathcal{F}}=O\left(\left(\frac{L^4K^2M{N_a^3}}{{N_r^{2}}{\tau_{\rm{c}}^3}}\right)^{1/4}+\frac{\sqrt{K{N_a}{{N_r}{\ln{\tau_{\rm{c}}}}\sigma^2}}}{{N_r}\sqrt{{\tau_{\rm{c}}}}}+\frac{\sqrt{N_a}L{\nu\sqrt{KM{\tau_{\rm{c}}}}}}{{N_r}\sqrt{\tau_{\rm{c}}}\omega\left(\nu\sqrt{M{\tau_{\rm{c}}}}\right)}\right).
\end{equation}
\end{thm}
\begin{proof}
Denoting $\bm{\Psi}=\frac{N_a}{N_r}\widetilde{\mathbf{Y}}$, we can write
\begin{equation}
\label{equ_XeSVD}
\begin{aligned}
    \left\|\widehat{\mathbf{X}}-\mathbf{X}\right\|_{\mathcal{F}}&=\left\|\bm{\Psi}\bm{\widehat\Pi}_K-\mathbf{X}\right
    \|_{\mathcal{F}}\\
    &\le \left\|\bm{\Psi}\bm{\Pi}_K-\mathbf{X}\right\|_{\mathcal{F}}+\left\|\bm{\Psi}\bm{\widehat\Pi}_K-\bm{\Psi}\bm{\Pi}_K\right\|_{\mathcal{F}}\\
    &\le \left\|\bm{\Psi}\bm{\Pi}_K-\mathbf{X}\right\|_{\mathcal{F}}+\left\|\bm{\Psi}\right\|_{\mathcal{F}}\left\|\bm{\widehat\Pi}_K-\bm{\Pi}_K\right\|_{\mathcal{F}}.
\end{aligned}
\end{equation}

First, according to Theorem 7 of \cite{dwork2014analyze}, if $\lambda_K^2-\lambda_{K+1}^2=\omega\left(\nu\sqrt{M{\tau_{\rm{c}}}}\right)$, then with high probability
\begin{equation}
    \left\|\bm{\widehat\Pi}_K-\bm{\Pi}_K\right\|_2=O\left(\frac{\nu\sqrt{M{\tau_{\rm{c}}}}}{\omega\left(\nu\sqrt{M{\tau_{\rm{c}}}}\right)}\right),
\end{equation}
and hence
\begin{equation}
\label{equ_pike}
    \left\|\bm{\widehat\Pi}_K-\bm{\Pi}_K\right\|_\mathcal{F}=O\left(\frac{\nu\sqrt{KM{\tau_{\rm{c}}}}}{\omega\left(\nu\sqrt{M{\tau_{\rm{c}}}}\right)}\right).
\end{equation}
Furthermore, we have
\begin{equation}
\label{equ_psie}
    \left\|\bm{\Psi}\right\|_{\mathcal{F}}=\frac{N_a}{N_r}\left\|\widetilde{\mathbf{Y}}\right\|_{\mathcal{F}}\le \frac{N_a}{N_r}\left\|{\mathbf{Y}}\right\|_{\mathcal{F}}\le \frac{N_a}{N_r}\sqrt{M}L.
\end{equation}

Next, using Theorem 1.1 in \cite{keshavan2010matrix}, with high probability, there exist constants $c_0$ and $c_1$ such that
\begin{equation}
\label{equ_psipi}
    \frac{1}{\sqrt{M{N_a}{\tau_{\rm{c}}}}}\left\|\bm{\Psi}\bm{\Pi}_K-\mathbf{X}\right\|_{\mathcal{F}}\le c_0{x_{\rm{max}}}\left(\frac{K^2M{N_a^3}}{{N_r}^{2}{\tau_{\rm{c}}^3}}\right)^{1/4}+c_1\frac{\sqrt{K{N_a}}}{{N_r}\sqrt{M{\tau_{\rm{c}}}}}\left\|\left(\widetilde{\mathbf{N}}\right)_\Omega\right\|_2,
\end{equation}
where $x_{\rm{max}}=\max_{(i,j)}\left|\mathbf{X}\left({i,j}\right)\right|$ and $\left(\widetilde{\mathbf{N}}\right)_\Omega$ denotes the matrix obtained from $\left({\mathbf{N}}\right)_\Omega$ after the trimming step. 
By using the Theorem 1.3 in \cite{keshavan2010matrix}, with high probability, there exists a constant $c_2$ such that
\begin{equation}
\label{equ_wne}   
    \left\|\left(\widetilde{\mathbf{N}}\right)_\Omega\right\|_2\le {c_2}{\sigma}\left({M{N_r}\log{\tau_{\rm{c}}}}\right)^{1/2}.
\end{equation}

Plugging (\ref{equ_pike})-(\ref{equ_wne}) and $x_{\rm{max}}\le L$ into (\ref{equ_XeSVD}), we obtain (\ref{equ_XerSVD}). $\hfill\Box$
\end{proof}

\begin{rem}
Since $L=O\left(\sqrt{\tau_{\rm{d}}}\right)$ and $\nu=O\left(\tau_{\rm{d}}\right)$, the sum of the first two terms in (\ref{equ_XerSVD}) scales with $\tau_{\rm{d}}$ as $O({\tau_{\rm{d}}^{-1/4}})$, and the third term in (\ref{equ_XerSVD}) represents the completion error caused by the perturbation noise, that scales as $O\left(\frac{\tau_{\rm{d}}^{3/2}}{\omega\left({\tau_{\rm{d}}^{3/2}}\right)}\right)$, which dominates the matrix completion error. Hence using (\ref{equ_CEE}) and (\ref{equ_mcem}), we arrive at the following corollary regarding the scaling of the channel estimation error by Algorithm 2 at AP $m$ with respect to the data payload size $\tau_{\rm{d}}$.
\end{rem}

\begin{coro}
For fixed privacy parameters $\epsilon$ and $\delta$, and fixed pilot length $\tau_{\rm{p}}$, the estimation error of the proposed privacy-preserving channel estimator that employs Algorithm 2 scales with the data payload size $\tau_{\rm{d}}$ as
\begin{equation}
\label{equ_ceeO}
    \left\|\widehat{\mathbf{H}}_m-\mathbf{H}_m\right\|_{\mathcal{F}}=O\left(\frac{\tau_{\rm{d}}^{3/2}}{\omega\left({\tau_{\rm{d}}^{3/2}}\right)}\right).
\end{equation}
\end{coro}

In summary, we see that for both Algorithms 1 and 2, the channel estimator error consists of a privacy independent component, that is due to the channel noise and matrix completion error, and a privacy-induced component, that is due to the perturbation noise. A higher privacy level leads to a higher privacy-induced channel estimation error, and vice versa. As the payload size increases, both error components decrease. However, for Algorithm 1, the channel estimation error is dominated by the privacy-independent component; whereas for Algorithm 2, it is dominated by the privacy-induced component.

\section{SIMULATION RESULTS}
\subsection{Simulation Setup}
We consider a cell-free massive MIMO system covering a hexagonal region with radius $R=1{\rm{km}}$, where APs and users are randomly and uniformly distributed. The channel model in (\ref{equ_channel}) is adopted to generate channel matrices $\mathbf{H}_m$ with the large-scale fading $\beta_{k,m}$ modeled as
\begin{equation}
    \beta_{k,m}=10^{-\frac{{\rm{PL}}\left(d_{k,m}\right)+{\sigma_{\rm{sh}}{z_{k,m}}}}{10}},
\end{equation}
where ${\rm{PL}}\left(d_{k,m}\right)\left({\rm{dB}}\right)$ is the path loss between AP $m$ and user $k$ with distance $d_{k,m}$; $\sigma_{\rm{sh}}\left({\rm{dB}}\right)$ is the standard deviation of shadow fading and ${z_{k,m}}\sim\mathcal{N}_{\rm{c}}\left(0,1\right)$.  $\tau_{\rm{p}}=K$ orthonormal pilot sequences are used, resulting in an orthonormal pilot matrix  $\mathbf{P}$. Data symbols are independently drawn from the QPSK constellation with unit average power. We consider two settings of the number of users: $K=5$ and $K=25$.
All simulation parameters are shown in Table I. The channel estimation performance is evaluated by the normalized mean squared error (NMSE) defined as
\begin{equation}
    {\rm{NMSE}}=\mathbb{E}\left\{\frac{\left\|\widehat{\mathbf{H}}-\mathbf{H}\right\|_{\mathcal{F}}^2}{\left\|\mathbf{H}\right\|_{\mathcal{F}}^2}\right\}.
\end{equation}
The data detection performance is evaluated by symbol error rate (SER). Both NMSE and SER are obtained through Monte-Carlo simulations with fixed large-scale fadings $\{\beta_{k,m}\}$ and a minimum of 500 independent fast channel realizations $\{\mathbf{g}_{k,m}\}$.
\begin{table}
  \caption{Basic Simulation Parameters}
  \centering
  \begin{tabular}{|c|c|c|}
  \hline
  Parameter & Meaning & Value\\
  \hline
  $M$   & The number of APs & 100\\
  \hline
  $K$   & The number of users & 5, 25\\
  \hline
  $N_a$ & The number of antennas at each AP & 4\\
  \hline
  $N_r$ & The number of RF chains at each AP & 2\\
  \hline
  $\sigma_{\rm{sh}}$  & The standard deviation of shadow fading & 8 dB\\
  \hline
  ${\rm{PL}}\left(d\right)$ & The path loss with distance $d\left({\rm{m}}\right)$ & $36.8+36.7\log_{10}\left(d\right)$ dB\\
  \hline
  $\sigma^2$ & The variance of received noise sample & $10^{-13}$ Watts\\
  \hline
\end{tabular}
\end{table}

In Algorithm 1, we approximate the rank constraint ${\rm{rank}}\left(\mathbf{X}\right)\le K$ with the nuclear norm constraint $\left\|\mathbf{X}\right\|_{\rm{nuc}}\le K$. However, the true $\left\|\mathbf{X}\right\|_{\rm{nuc}}$ is far smaller than $K$ due to the large-scale fading. Hence, in our implementation, we replace the rank bound $K$ in (\ref{equ_XmFWPP}) with an appropriate nuclear norm bound $r$.
According to (\ref{equ_Hmap}) and (\ref{equ_Sap}), for massive MIMO, we have
\begin{equation}
    \left\|\mathbf{X}\right\|_{\mathcal{F}}=\left\|\mathbf{H}\mathbf{S}\right\|_{\mathcal{F}}\le \left\|\mathbf{H}\right\|_{\mathcal{F}}\left\|\mathbf{S}\right\|_{\mathcal{F}}=\sqrt{K{\tau_{\rm{c}}}{N_a}\sum\nolimits_{m=1}^M\sum\nolimits_{k=1}^K{\beta_{k,m}}}.
\end{equation}
Since $\left\|\mathbf{X}\right\|_{\rm{nuc}}\le \sqrt{{\rm{rank}}\left(\mathbf{X}\right)}\left\|\mathbf{X}\right\|_{\mathcal{F}}$\cite{XPengconnec}, we can bound $\left\|\mathbf{X}\right\|_{\rm{nuc}}$ by
\begin{equation}
\label{equ_rbound}
    r=\sqrt{K^2{\tau_{\rm{c}}}{N_a}\sum\nolimits_{m=1}^M\sum\nolimits_{k=1}^K{\beta_{k,m}}}.
\end{equation}
For $K=5$, we choose $r\in[0.001,0.01]$ by cross-validation, e.g., we uniformly choose 10 values of $r$ and run Algorithm 1 using them. The $r$ value that has the lowest NMSE is then chosen for the simulations. 
For $K=25$, we choose $r\in[0.01,0.1]$ by cross-validation. Moreover, for the number of iterations $T$, we choose  $T\in \left\{K,2K,\cdots,5K\right\}$ by cross-validation. 

For comparison, we consider the following three channel estimators:

(1) Non-private FW (NPFW): To show the performance upper bound of Algorithm 1 when privacy is not considered, we set $\mu=0$ and $T=200$.

(2) Non-private SVD (NPSVD): To show the performance upper bound of Algorithm 2 when privacy is not considered, we set $\nu=0$.

(3)  Pilot-only (PO): We also consider the pilot-only method, where each AP $m$ estimates its channel matrix $\mathbf{H}_m$ locally based on its received pilots $\{\mathbf{Y}_m(:, t), t \in {\cal T}_p \}$ only. 
Specifically, each AP $m$ first computes its least squares (LS) estimate as
\begin{equation}
   \widehat{\mathbf{H}}_m=\mathbf{Y}_m(:, 1:\tau_{\rm{p}}){\mathbf{P}}^{\rm{H}}.
\end{equation}
Then the local linear minimum mean-squared error (LMMSE) estimate of data symbols ${\mathbf{D}}_m$ is computed as follows. Denoting $\mathbf{y}_m[t]$ as the $N_r$-dimensional vector consisting of the non-zero elements of $\mathbf{Y}_m\left(:,t\right)$. Then from (\ref{equ_YmR}) we can write $\mathbf{y}_m\left[t\right]=\mathbf{C}_m\left[t\right]\mathbf{r}_m\left[t\right]$, where $\mathbf{C}_m\left[t\right]\in\mathbb{C}^{{N_r}\times{N_a}}$ and $\mathbf{C}_m\left[t\right]\left(i,j\right)=1$ if the $j$-th antenna is connected to the $i$-th RF chain, and it is 0 otherwise. Then LMMSE estimate is given by
\begin{equation}
    \widehat{\mathbf{D}}_m\left[t\right]=\left(\mathbf{F}_m^{\rm{H}}\mathbf{F}_m+\sigma^2{\mathbf{I}_{K}}\right)^{-1}\mathbf{F}_m^{\rm{H}}\mathbf{y}_m\left[t\right], \tau_{\rm{p}}+1\le t\le \tau_{\rm{c}},
\end{equation}
with $\mathbf{F}_m=\mathbf{C}_m\left[t\right]\widehat{\mathbf{H}}_m$. At last, $\left\{\widehat{\mathbf{D}}_m\right\}$ are sent to the CPU which performs data detection based on the combined statistic according to (\ref{equ_dtc}). Since the PO method does not send any private signal to the CPU, it is perfectly privacy-preserving.

\subsection{Results}

Fig. \ref{fig:vsEpsNMSEK5} and \ref{fig:vsEpsSERK5} respectively show the NMSE performance of channel estimation and the corresponding SER performance of data detection versus the privacy parameter $\epsilon$ with $K=5$, $\tau_{\rm{c}}=100$, and $\delta=0.1$ for different methods. It can be seen that the performance of both Algorithm 1 and 2 improve as $\epsilon$ increases, which means the privacy level degrades. Algorithm 1 significantly outperforms Algorithm 2 under both private and non-private cases. Moreover, despite of the added perturbation noise to achieve privacy, both algorithms outperform the PO method, by exploiting the received data payload signal. 

Fig. \ref{fig:vsTCNMSEK5} and \ref{fig:vsTCSERK5} respectively show the NMSE performance of channel estimation and the corresponding SER performance of data detection versus the payload size $\tau_{\rm{d}}$ with $K=5$, fixed privacy parameter $\epsilon=1$ and $\delta=0.1$ for different methods. Since the PO method makes use of the received pilot signal only for channel estimation, its performance remains the same as $\tau_{\rm{d}}$ increases. On the contrary, the performances of Algorithm 1 and 2 improve as $\tau_{\rm{d}}$ increases. Hence, the proposed methods can utilize the received data payload signal to improve the accuracy of channel estimation and data detection, while maintaining the same privacy level at the same time. Algorithm 1 significantly outperforms Algorithm 2 under both private and non-private cases for all payload size.

For $K=25$, we also show the performances of five methods versus privacy parameter $\epsilon$ and payload size $\tau_d$ in Fig.
\ref{fig:vsEpsK25} and \ref{fig:vsTCK25}, respectively. The performances of both Algorithm 1 and Algorithm 2 still improve as $\epsilon$ or $\tau_{\rm{d}}$ increases. Algorithm 1 still significantly outperforms Algorithm 2 and the PO method for all considered $\epsilon$ and $\tau_{\rm{d}}$. However, Algorithm 2 is less effective when the number of users is high.

\begin{figure}[H]
\centering
\subfigure{
\label{fig:vsEpsNMSEK5}
\includegraphics[width=0.45\textwidth]{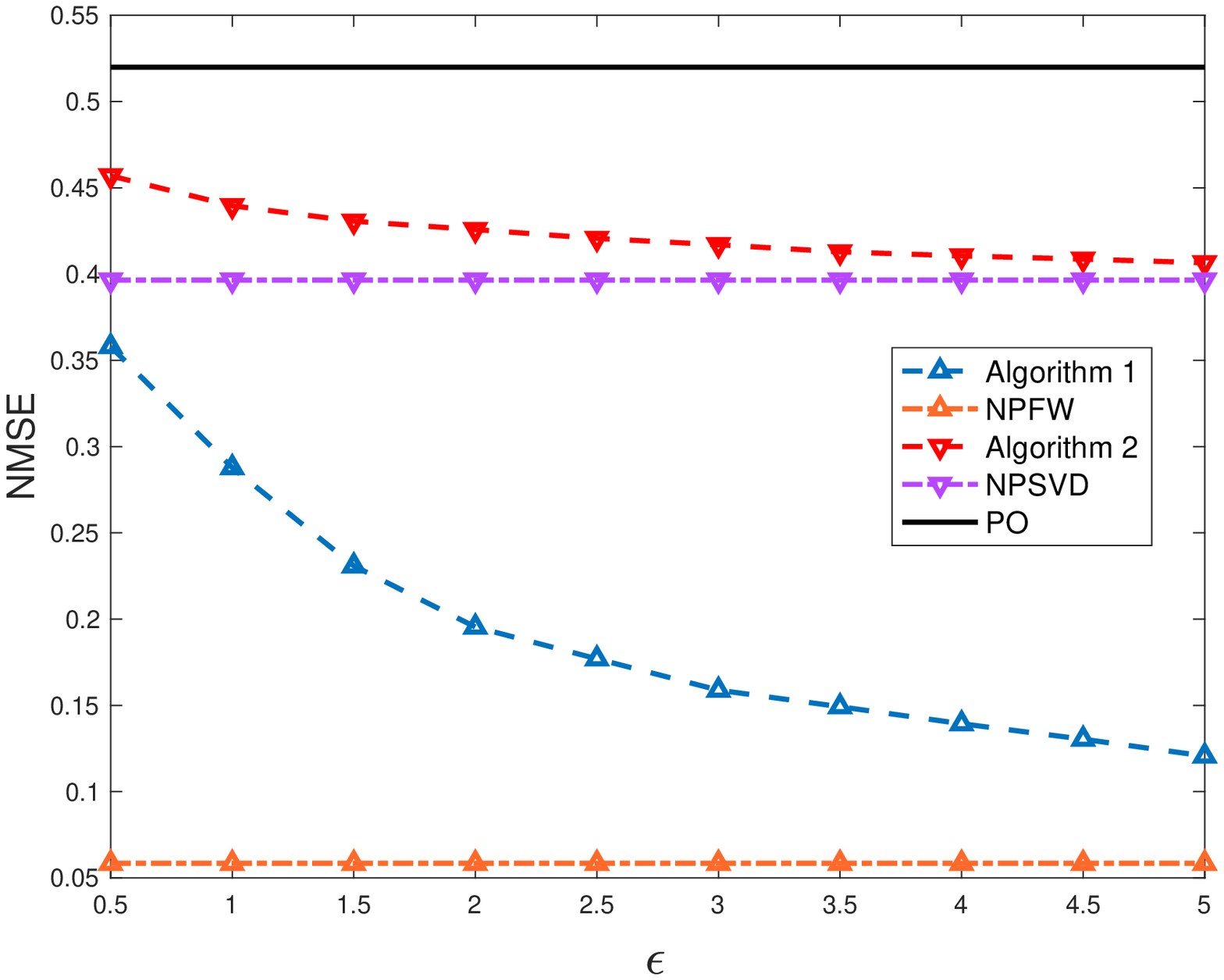}}
\subfigure{
\label{fig:vsEpsSERK5}
\includegraphics[width=0.45\textwidth]{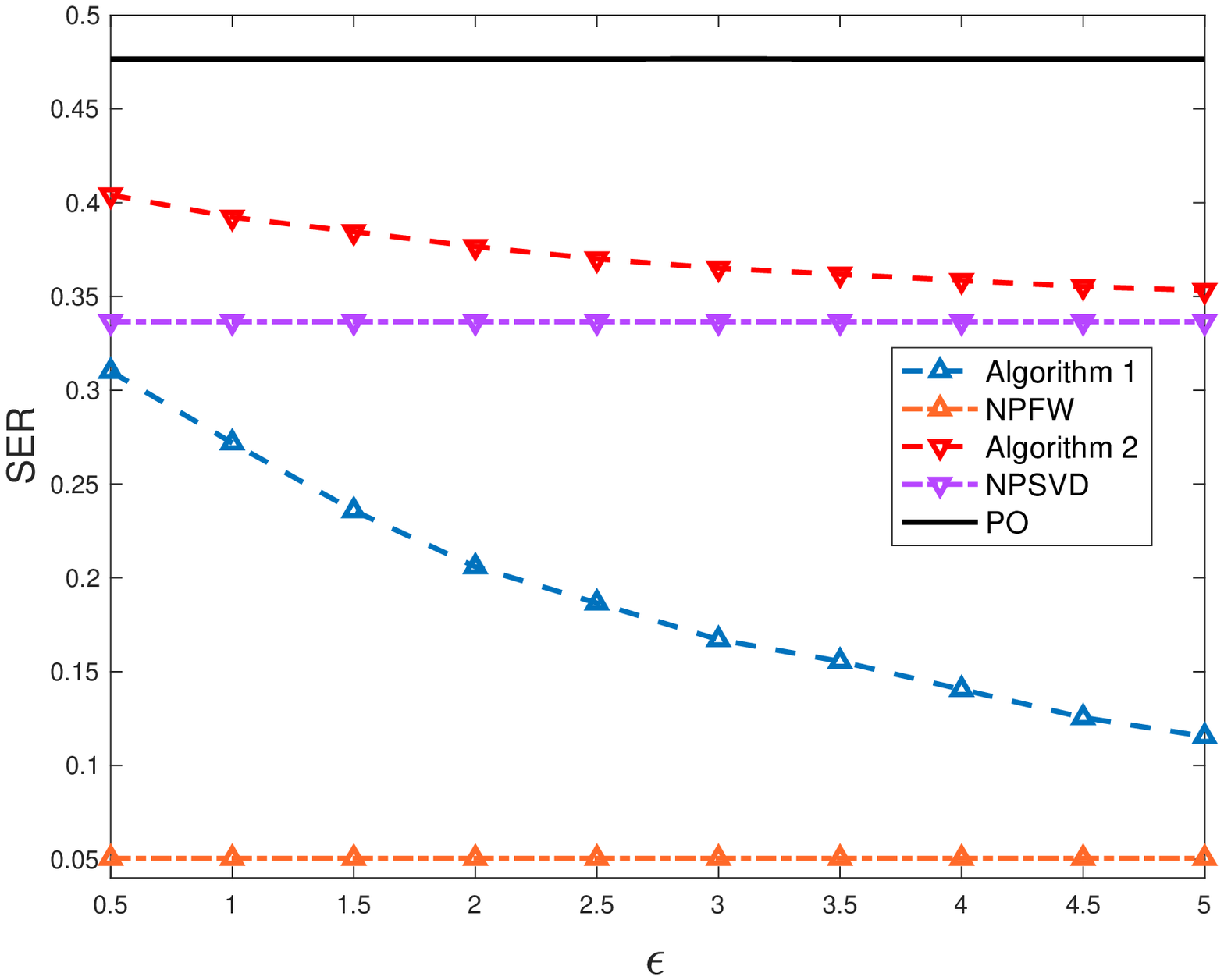}}
\caption{The performance versus $\epsilon$ with $K=5$, $\tau_{\rm{c}}=100$ and $\delta=0.1$. (a) NMSE of channel estimation. (b) SER of data detection.}
\label{fig:vsEpsK5}
\end{figure}

\begin{figure}[H]
\centering
\subfigure{
\label{fig:vsTCNMSEK5}
\includegraphics[width=0.45\textwidth]{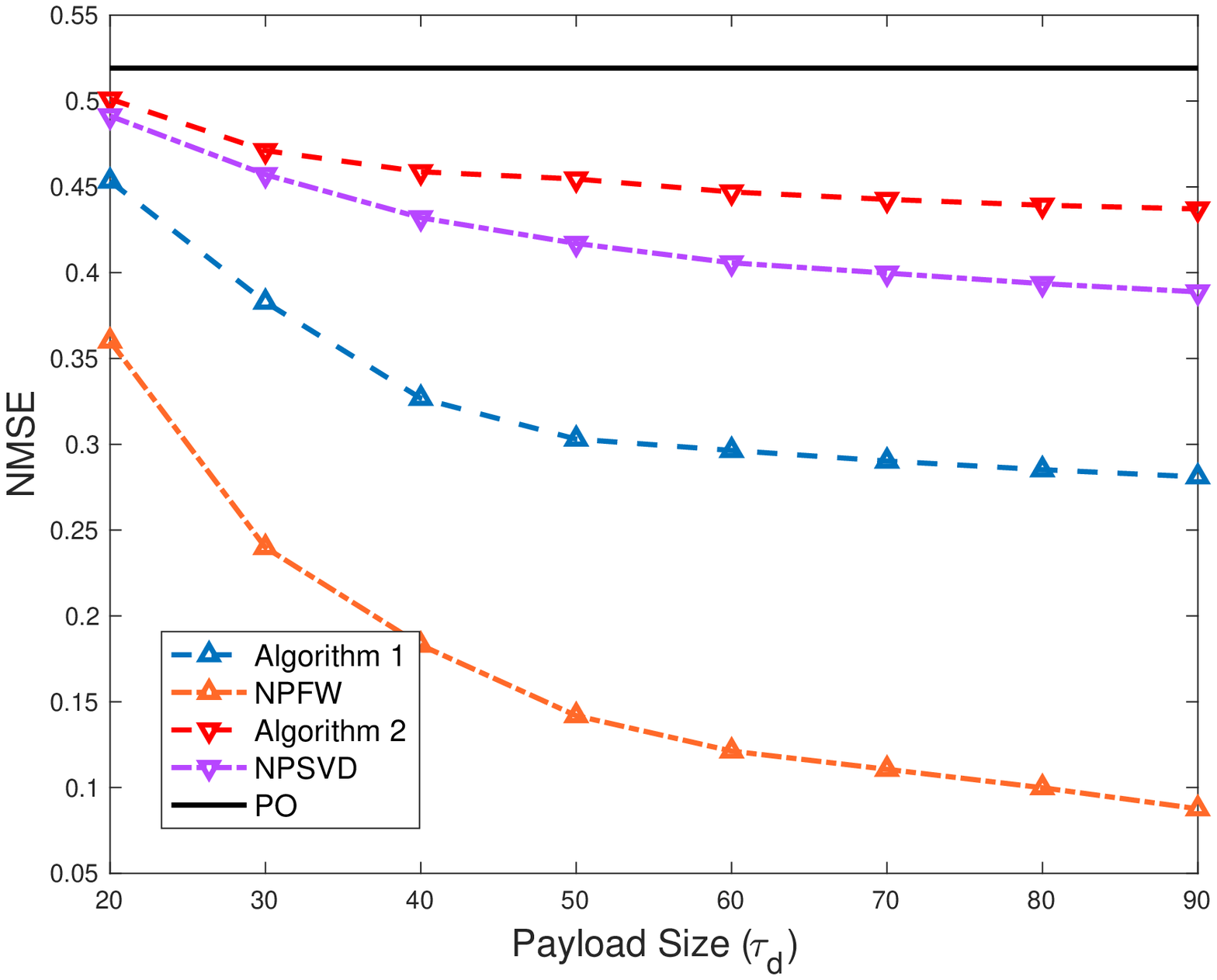}}
\subfigure{
\label{fig:vsTCSERK5}
\includegraphics[width=0.45\textwidth]{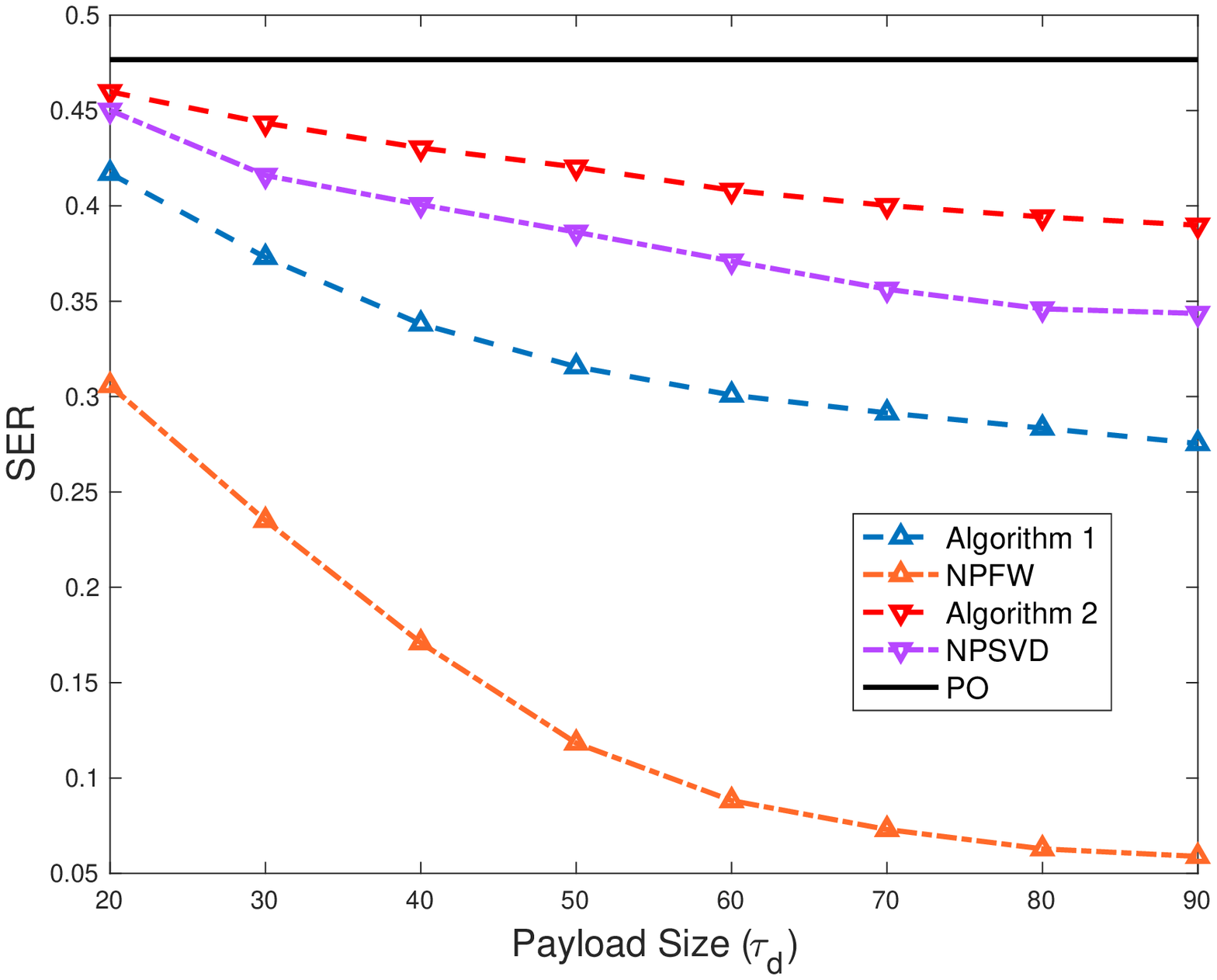}}
\caption{The performance versus payload size $\tau_{\rm{d}}$ with $K=5$, $\epsilon=1$ and $\delta=0.1$. (a) NMSE of channel estimation. (b) SER of data detection.}
\label{fig:vsTCK5}
\end{figure}

\begin{figure}[H]
\centering
\subfigure{
\label{fig:vsEpsNMSEK25}
\includegraphics[width=0.45\textwidth]{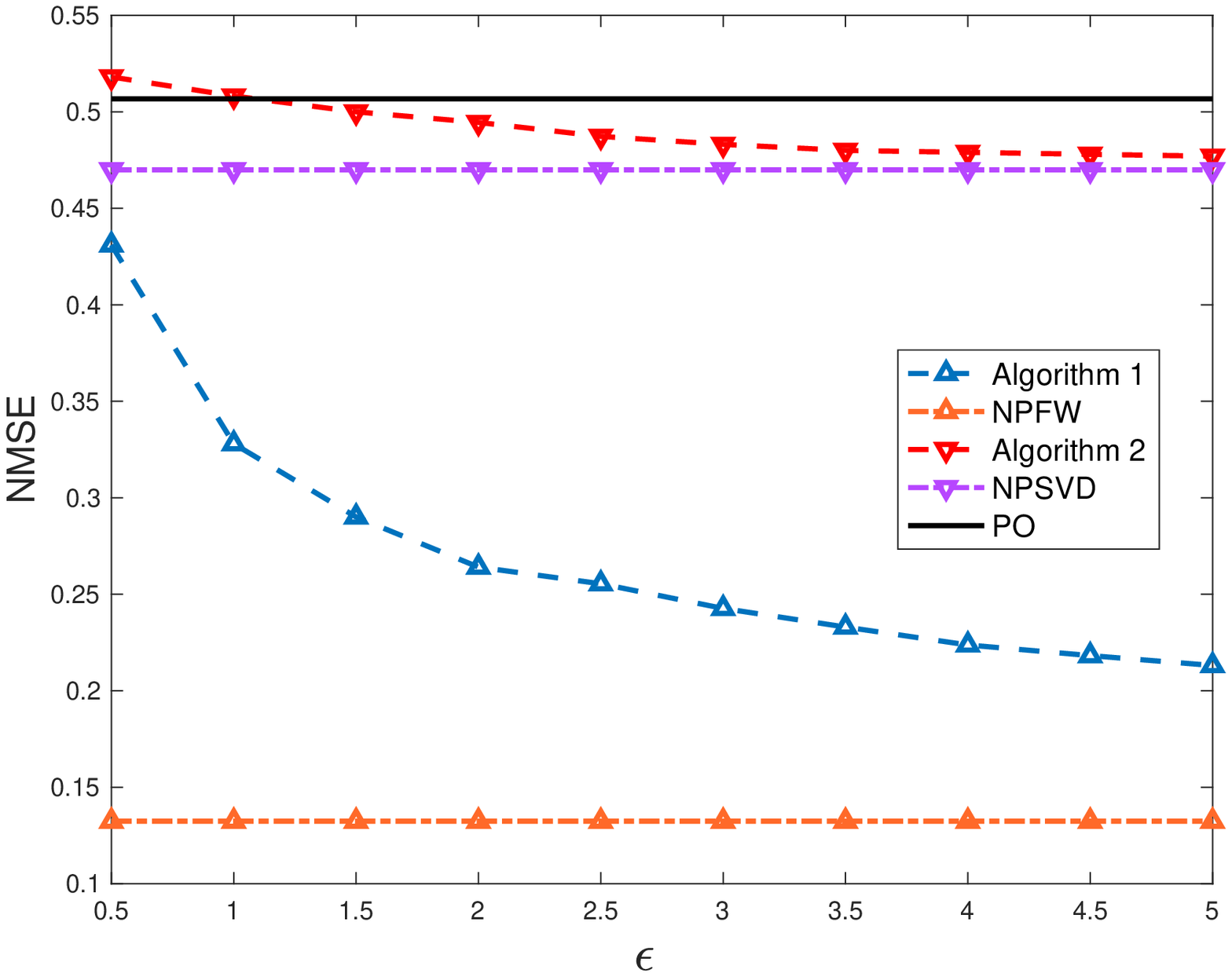}}
\subfigure{
\label{fig:vsEpsSERK25}
\includegraphics[width=0.45\textwidth]{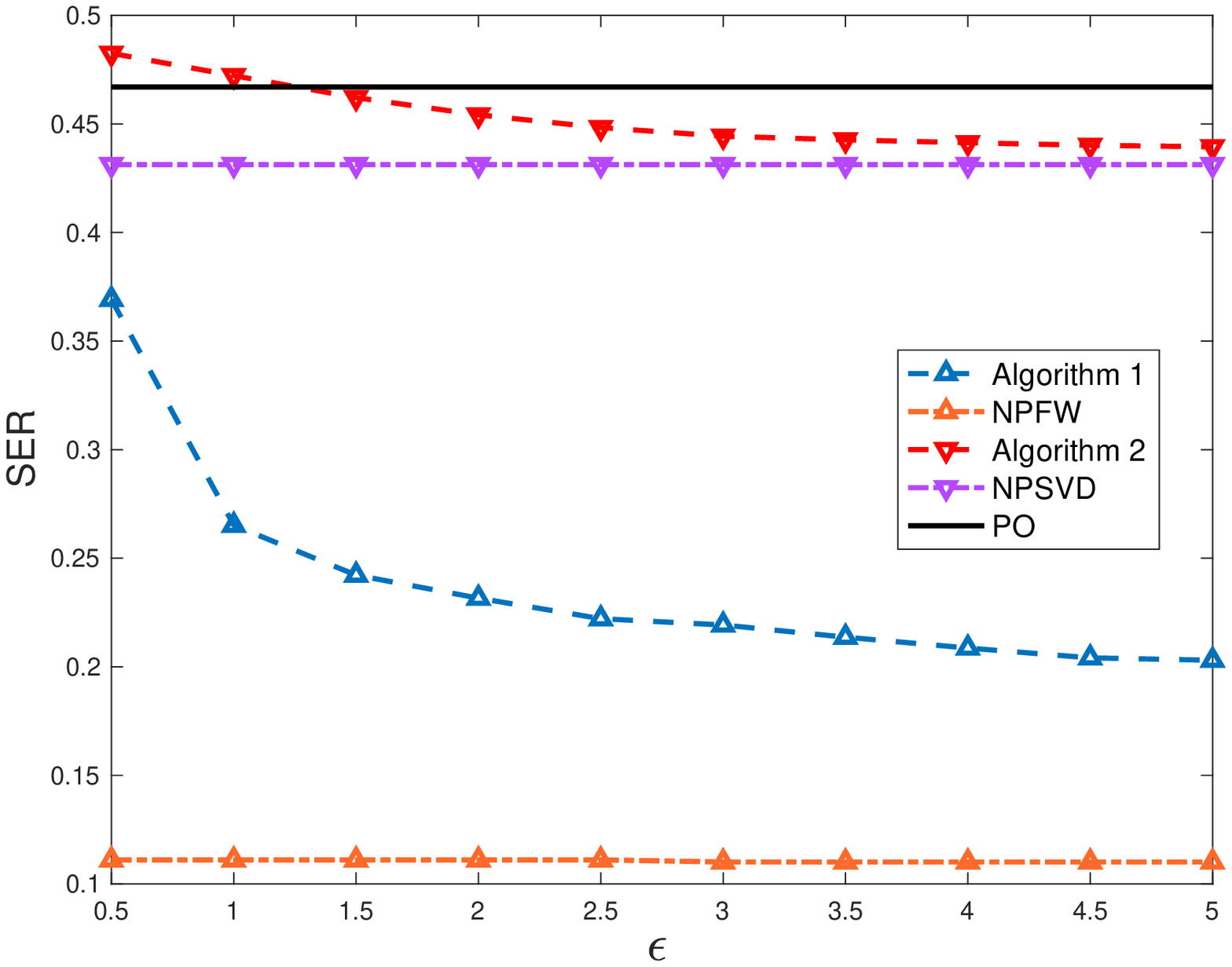}}
\caption{The performance versus $\epsilon$ with $K=25$, $\tau_{\rm{c}}=100$ and $\delta=0.1$. (a) NMSE of channel estimation. (b) SER of data detection.}
\label{fig:vsEpsK25}
\end{figure}

\section{Conclusions}
This paper considers a cell-free hybrid massive MIMO system, where the number of users is typically much smaller than the total number of antennas. Efficient uplink channel estimation and data detection with reduced number of pilots can be performed based on low-rank matrix completion. However, such a scheme requires the CPU to collect received signals from all APs, which may enable the CPU to infer the private information of user locations. To solve this problem, we develop and analyze privacy-preserving channel estimation schemes under the framework of differential privacy. The key ingredient of such a channel estimator is a joint differentially private noisy matrix completion algorithm, which consists of a global component implemented at the CPU and local components implemented at APs. Two joint differentially private channel estimators based respectively on FW and SVD are proposed and analyzed. In particular, we have shown that for both algorithms the estimation error can be mitigated while maintaining the same privacy level by increasing the payload size with fixed pilot size; and the scaling laws of both the privacy-induced and privacy-independent error components in terms of payload size are characterized. Simulation results corroborate the theoretical analysis and clearly demonstrate the tradeoff between privacy and channel estimation performance.

\begin{figure}[H]
\centering
\subfigure{
\label{fig:vsTCNMSEK25}
\includegraphics[width=0.45\textwidth]{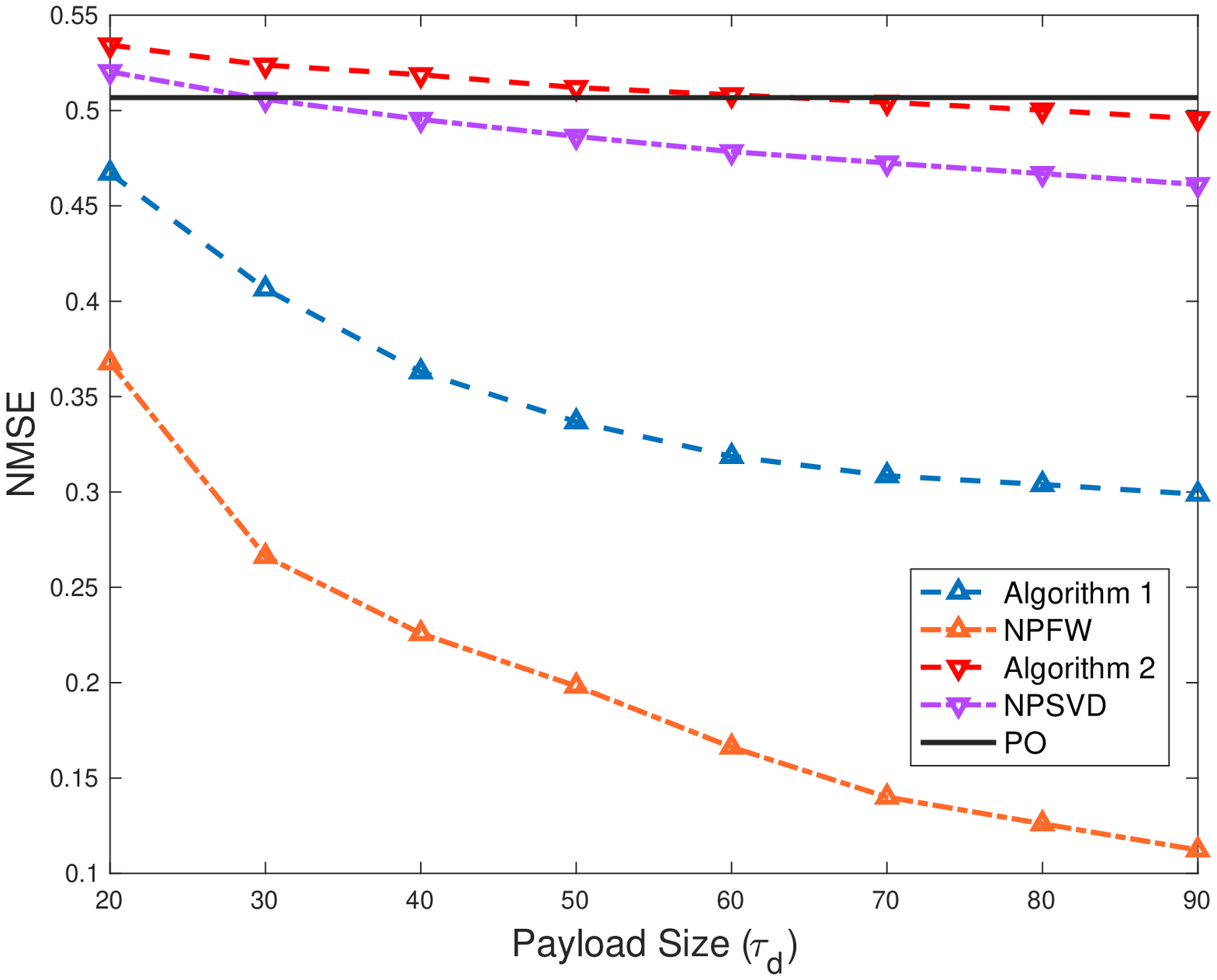}}
\subfigure{
\label{fig:vsTCSERK25}
\includegraphics[width=0.45\textwidth]{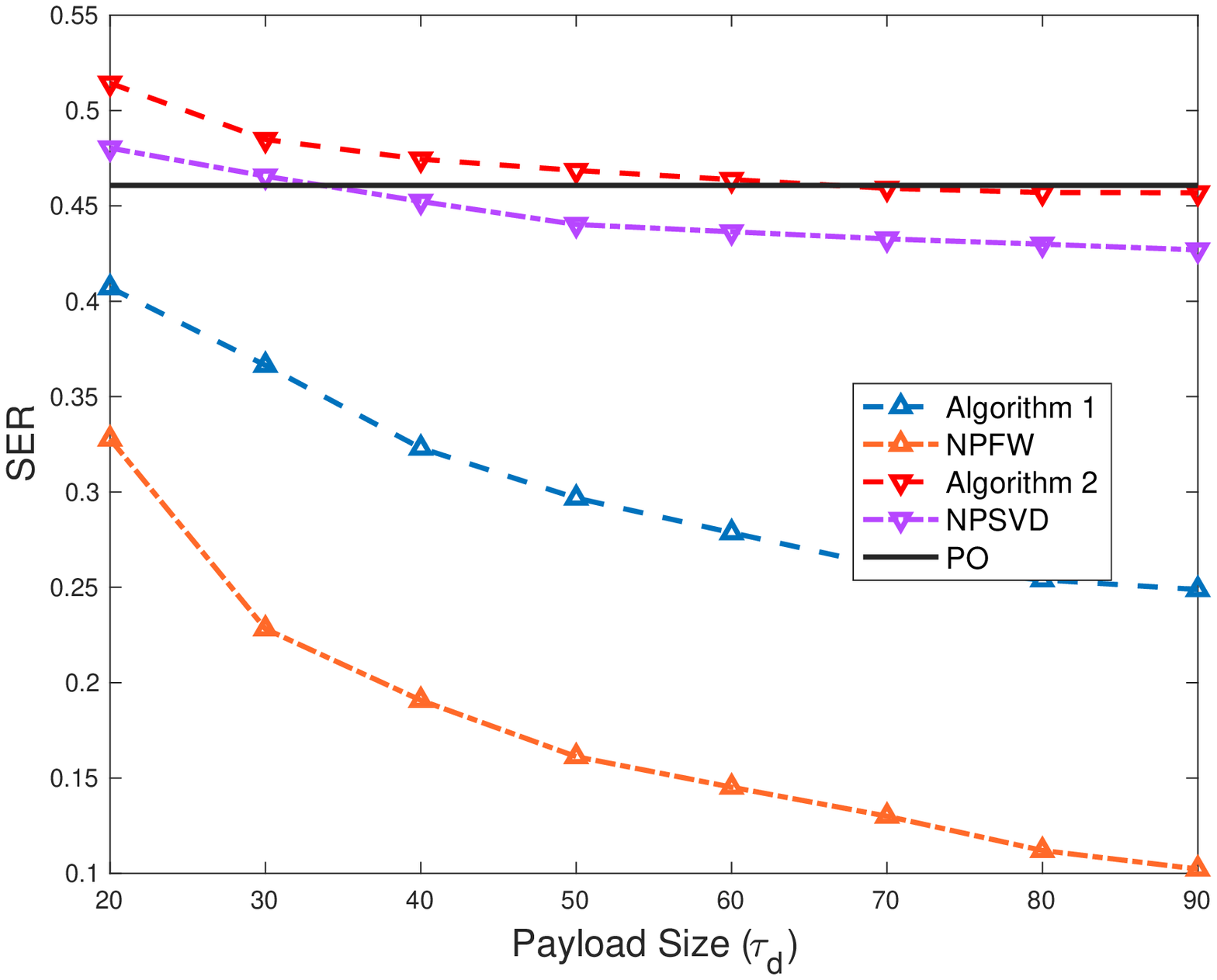}}
\caption{The performance versus payload size $\tau_{\rm{d}}$ with $K=25$, $\epsilon=1$ and $\delta=0.1$. (a) NMSE of channel estimation. (b) SER of data detection.}
\label{fig:vsTCK25}
\end{figure}

\appendix
\appendices

\subsection{Proof of Lemma \ref{lem_ugFW}}
Define a function $\Gamma_{\Omega}\left(\mathbf{X}\right)=\frac{1}{2\left|\Omega\right|}\left\|\left(\mathbf{X}\right)_{\Omega}-\mathbf{Y}\right\|_{\mathcal{F}}^2$ on the feasible set $\mathcal{D}:\left\|\mathbf{X}\right\|_{\rm{nuc}}\le K$. The curvature parameter $C_{\Gamma}$ of the above function can be defined as
\begin{equation}
\label{equ_Cgamma}
    C_{\Gamma}=\max_{\mathbf{X}_a,\mathbf{S}\in\mathcal{D}}\frac{2}{\kappa^2}\left(\Gamma_{\Omega}\left(\mathbf{X}_b\right)-\Gamma_{\Omega}\left(\mathbf{X}_a\right)-\left<\mathbf{X}_b-\mathbf{X}_a, \nabla\Gamma_{\Omega}\left(\mathbf{X}_a\right)\right>_{\mathcal{F}}\right),
\end{equation}
where $\kappa\in[0,1]$; $\mathbf{X}_b=\mathbf{X}_a+\kappa\left(\mathbf{S}-\mathbf{X}_a\right)$; $\nabla\Gamma_{\Omega}\left(\mathbf{X}_a\right)=\frac{1}{\left|\Omega\right|}\left(\left(\mathbf{X}_a\right)_{\Omega}-\mathbf{Y}\right)$ is the gradient of $\Gamma_{\Omega}\left(\mathbf{X}\right)$ at $\mathbf{X}_a$. It then follows from the definition in (\ref{equ_Cgamma}) that for any $\mathbf{X}_a$ and $\mathbf{S}$
\begin{equation}
\label{equ_sf}
    \Gamma_{\Omega}\left(\mathbf{X}_b\right)\le \Gamma_{\Omega}\left(\mathbf{X}_a\right)+\left<\mathbf{X}_b-\mathbf{X}_a, \nabla\Gamma_{\Omega}\left(\mathbf{X}_a\right)\right>_{\mathcal{F}}+\frac{C_\Gamma{\kappa^2}}{2}.
\end{equation}

According to (\ref{equ_XupW}), we have
\begin{equation}
    \mathbf{X}^{(n)}=\mathbf{X}^{(n-1)}+\eta^{(n)}\left(\mathbf{W}^{(n)}-\mathbf{X}^{(n-1)}\right).
\end{equation}
Recall that $\left\|{\mathbf{W}}^{(n)}\right\|_{\rm{nuc}}\le K, \forall n$, $\eta^{(1)}=1$ and $\eta^{(n)}=\frac{1}{T}, n=2,\cdots,T$, thus we have $\left\|{\mathbf{X}}^{(n)}\right\|_{\rm{nuc}}\le K, \forall n$. By letting $\mathbf{X}_b=\mathbf{X}^{(n)}$, $\mathbf{X}_a=\mathbf{X}^{(n-1)}$, $\mathbf{S}=\mathbf{W}^{(n)}$ and $\kappa=\eta^{(n)}$ in (\ref{equ_sf}), we have 
\begin{equation}
\label{equ_Gammayn}
\begin{aligned}
    \Gamma_{\Omega}\left(\mathbf{X}^{(n)}\right)\le \Gamma_{\Omega}\left(\mathbf{X}^{(n-1)}\right)+\eta^{(n)}\left<\mathbf{W}^{(n)}-\mathbf{X}^{(n-1)}, \nabla\Gamma_{\Omega}\left(\mathbf{X}^{(n-1)}\right)\right>_{\mathcal{F}}+\frac{C_\Gamma{\left(\eta^{(n)}\right)^2}}{2}.
\end{aligned}
\end{equation}

Note that $\nabla\Gamma_{\Omega}\left(\mathbf{X}^{(n-1)}\right)=\frac{1}{\left|\Omega\right|}\mathbf{J}^{(n-1)}$. If $\mathbf{W}^{(n)}$ satisfies (\ref{equ_gammaC}), we have 
\begin{equation}
\label{equ_gammaapprox}
    \left<\mathbf{W}^{(n)}-\mathbf{X}^{(n-1)},\nabla\Gamma_{\Omega}\left(\mathbf{X}^{(n-1)}\right)\right>_{\mathcal{F}}\le \left<\mathbf{O}^{(n)}-\mathbf{X}^{(n-1)},\nabla\Gamma_{\Omega}\left(\mathbf{X}^{(n-1)}\right)\right>_{\mathcal{F}}+\gamma.
\end{equation}
According to \cite{jain2017differentially}, $\mathbf{O}^{(n)}=\arg\min_{\left\|\mathbf{O}\right\|_{\rm{nuc}}\le K} \left<\mathbf{O},\nabla{\Gamma_\Omega}\left(\mathbf{X}^{(n-1)}\right)\right>_{\mathcal{F}}$. Therefore, we have $\mathbf{O}^{(n)}=\arg\max_{\left\|\mathbf{O}\right\|_{\rm{nuc}}\le K}\left<\mathbf{X}^{(n-1)}-\mathbf{O},\nabla{\Gamma_\Omega}\left(\mathbf{X}^{(n-1)}\right)\right>_{\mathcal{F}}$. We define $\Upsilon\left(\Theta\right) =\Gamma_\Omega\left(\Theta\right)-\Gamma_\Omega\left(\widehat{\mathbf{X}}\right)$, where $\widehat{\mathbf{X}}$ is given by (\ref{equ_Pnorm}). The convexity of $\Gamma_{\Omega}\left(\mathbf{X}\right)$ implies \cite{jaggi2013revisiting}
\begin{equation}
\label{equ_conve}
    \left<\mathbf{X}^{(n-1)}-\mathbf{O}^{(n)},\nabla{\Gamma_\Omega}\left(\mathbf{X}^{(n-1)}\right)\right>_{\mathcal{F}}\ge \Upsilon\left(\mathbf{X}^{(n-1)}\right).
\end{equation}
Plugging (\ref{equ_gammaapprox}) and (\ref{equ_conve}) into (\ref{equ_Gammayn}), we have
\begin{equation}
\label{equ_Gammaynr}
\begin{aligned}
    \Gamma_{\Omega}\left(\mathbf{X}^{(n)}\right)\le \Gamma_{\Omega}\left(\mathbf{X}^{(n-1)}\right)-\eta^{(n)} \Upsilon\left(\mathbf{X}^{(n-1)}\right)+\eta^{(n)}\left(\gamma+\frac{C_\Gamma{\eta^{(n)}}}{2}\right).
\end{aligned}
\end{equation}
Letting $n=T$ and subtracting $\Gamma_\Omega\left(\widehat{\mathbf{X}}\right)$ from both sides, we have
\begin{equation}
\label{equ_xrb}
\begin{aligned}
    \Upsilon\left(\mathbf{X}^{(T)}\right)&\le \Upsilon\left(\mathbf{X}^{(T-1)}\right)-\eta^{(T)}\Upsilon\left(\mathbf{X}^{(T-1)}\right)+\eta^{(T)}\left(\gamma+\frac{C_{\Gamma}\eta^{(T)}}{2}\right)\\
    &= \left(1-\eta^{(T)}\right)\Upsilon\left(\mathbf{X}^{(T-1)}\right)+\eta^{(T)}\left(\gamma+\frac{C_{\Gamma}\eta^{(T)}}{2}\right)\\
    &= \sum\limits_{n=T}^1{\left(\prod_{j=n+1}^{T}\left(1-\eta^{(j)}\right)\right)\eta^{(n)}\left(\gamma+\frac{C_{\Gamma}\eta^{(n)}}{2}\right)}+\prod_{n=1}^{T}\left(1-\eta^{(n)}\right)\Upsilon\left(\mathbf{X}^{(0)}\right).
\end{aligned}
\end{equation}
Recall that $\eta^{(1)}=1$ and $\eta^{(n)}=\frac{1}{T}, n=2,\cdots,T$, thus we have $\prod_{n=1}^{T}\left(1-\eta^{(n)}\right)=0$ and $0\le\prod_{j=n+1}^{T}\left(1-\eta^{(j)}\right)\le 1, n=1,\cdots, T$. Then (\ref{equ_xrb}) can be written as
\begin{equation}
    \Upsilon\left(\mathbf{X}^{(T)}\right)\le \sum\limits_{n=1}^T\eta^{(n)}\left(\gamma+\frac{C_{\Gamma}\eta^{(n)}}{2}\right).
\end{equation}
Recall that $\Upsilon\left(\mathbf{X}^{(T)}\right)=\Gamma_\Omega\left(\mathbf{X}^{(T)}\right)-\Gamma_\Omega\left(\widehat{\mathbf{X}}\right)$, we have
\begin{equation}
\label{equ_gammaxk}
\begin{aligned}
    \Gamma_\Omega\left(\mathbf{X}^{(T)}\right)&\le \gamma+\frac{C_{\Gamma}}{2}+ \frac{T-1}{T}\left(\gamma+\frac{C_{\Gamma}}{2T}\right)+\Gamma_\Omega\left(\widehat{\mathbf{X}}\right)\\
    &\le 2\gamma+\frac{C_{\Gamma}}{2}+ \frac{C_{\Gamma}}{2T}+\Gamma_\Omega\left(\widehat{\mathbf{X}}\right).
\end{aligned}
\end{equation}
Note that
\begin{equation}
\begin{aligned}
\label{equ_gammahx}
\Gamma_\Omega\left(\widehat{\mathbf{X}}\right)&=\frac{1}{2\left|\Omega\right|}\left\|\left(\widehat{\mathbf{X}}\right)_{\Omega}-\mathbf{Y}\right\|_{\mathcal{F}}^2\\
&\buildrel{(a)}\over\le\frac{1}{2\left|\Omega\right|}\left\|\left(\mathbf{X}\right)_{\Omega}-\mathbf{Y}\right\|_{\mathcal{F}}^2\\
&\buildrel{(b)}\over=\frac{1}{2\left|\Omega\right|}\left\|\left(\mathbf{N}\right)_\Omega\right\|_{\mathcal{F}}^2\buildrel{(c)}\over\rightarrow \frac{\sigma^2}{2},
\end{aligned}
\end{equation}
where (a) is due to (\ref{equ_Pnorm}); (b) is because $\mathbf{Y}=\left(\mathbf{X}\right)_\Omega+\left(\mathbf{N}\right)_\Omega$; (c) is satisfied when $\left|\Omega\right|=M{N_r}{\tau_{\rm{c}}}$ is large, which holds in massive MIMO. Plugging (\ref{equ_gammahx}) into (\ref{equ_gammaxk}), we have
\begin{equation}
    \Gamma_\Omega\left(\mathbf{X}^{(T)}\right)\le 2\gamma+\frac{C_{\Gamma}}{2}+ \frac{C_{\Gamma}}{2T}+\frac{\sigma^2}{2}.
\end{equation}
Note that $C_{\Gamma}$ is upper bounded by $ \frac{K^2}{\left|\Omega\right|}$\cite{clarkson2010coresets}. Hence, we obtain (\ref{equ_xkee}).

\subsection{Proof of Lemma \ref{lem_diffe}}
First, we introduce the following lemma.
\begin{lem}[Theorem 4 of \cite{dwork2014analyze}]
\label{lem7}
    Let $\mathbf{v}\in\mathbb{C}^{n\times 1}$ be the largest right singular vector of matrix $\mathbf{A}\in\mathbb{C}^{m\times n}$ ($m>n$) and let $\mathbf{\widehat{v}}$ be the largest eigenvector of matrix $\mathbf{B}=\mathbf{A}^{\rm{H}}\mathbf{A}+\mathbf{C}$, where $\mathbf{C}\in\mathbb{C}^{n\times n}$ is a Hermitian matrix whose upper triangular and diagonal elements are i.i.d samples from $\mathcal{N}_{\rm{c}}\left(0,\sigma^2\right)$ and $\mathcal{N}\left(0,\sigma^2\right)$ respectively. Then with high probability
    \begin{equation}
        \left\|\mathbf{A}\mathbf{\widehat{v}}\right\|_{\mathcal{F}}^2\ge \left\|\mathbf{A}\mathbf{{v}}\right\|_{\mathcal{F}}^2-O\left(\sigma\sqrt{n}\right).
    \end{equation}
\end{lem}

Recall that $\mathbf{v}^{(n)}$ and $\mathbf{\widehat{v}}^{(n)}$ are respectively the largest right singular vector and eigenvector of $\mathbf{J}^{(n-1)}$ and $\mathbb{\widehat{W}}^{(n-1)}=\sum\limits_{m=1}^M{\mathbb{\widehat{J}}_m^{(n-1)}}={\left(\mathbf{J}^{(n-1)}\right)^{\rm{H}}\mathbf{J}^{(n-1)}}+\sum\limits_{m=1}^M\mathbf{G}_m^{(n-1)}$, where $\sum\limits_{m=1}^M\mathbf{G}_m^{(n)}$ is a Hermitian matrix whose upper triangular and diagonal elements are i.i.d samples from $\mathcal{N}_{\rm{c}}\left(0,M\mu^2\right)$ and $\mathcal{N}\left(0,M\mu^2\right)$ respectively.
According to Lemma \ref{lem7}, with high probability, we have
\begin{equation}
    \left\|\mathbf{J}^{(n-1)}\mathbf{\widehat{v}}^{(n)}\right\|_{\mathcal{F}}^2\ge \left\|\mathbf{J}^{(n-1)}\mathbf{v}^{(n)}\right\|_{\mathcal{F}}^2-O\left(\mu\sqrt{M\tau_{\rm{c}}}\right).
\end{equation}

We define ${\alpha}=\left<\mathbf{O}^{(n)},\frac{1}{\left|\Omega\right|}\mathbf{J}^{(n-1)}\right>_{\mathcal{F}}$ and $\widehat{\alpha}=\left<\mathbf{Q}^{(n)},\frac{1}{\left|\Omega\right|}\mathbf{J}^{(n-1)}\right>_{\mathcal{F}}$.
Then with high probability, the following holds
\begin{equation}
\begin{aligned}
    \widehat{\alpha}&= \left<\mathbf{Q}^{(n)},\frac{1}{\left|\Omega\right|}\mathbf{J}^{(n-1)}\right>_{\mathcal{F}} \buildrel{(a)}\over=-\frac{K\left\|\mathbf{J}^{(n-1)}\mathbf{\widehat{v}}^{(n)}\right\|_{\mathcal{F}}^2}{{\widetilde\lambda}^{(n)}{\left|\Omega\right|}}\\
    &\le -\frac{K\left(\left\|\mathbf{J}^{(n-1)}\mathbf{v}^{(n)}\right\|_{\mathcal{F}}^2-O\left(\mu\sqrt{M{\tau_{\rm{c}}}}\right)\right)}{{\widetilde\lambda^{(n)}}{\left|\Omega\right|}}\\
    &\buildrel{(b)}\over=\frac{K\left(\frac{\lambda^{(n)}\left|\Omega\right|}{K}\alpha+O\left(\mu\sqrt{M{\tau_{\rm{c}}}}\right)\right)}{{\widetilde\lambda^{(n)}}{\left|\Omega\right|}}
\end{aligned}
\end{equation}
where (a) follows from the Frobenius inner product; (b) follows from the definition of $\alpha$. Then we have 
\begin{equation}
\begin{aligned}
\label{equ_diff}
    \widehat{\alpha}-\alpha&\le\left(\frac{\lambda^{(n)}}{\widetilde{\lambda}^{(n)}}-1\right)\alpha+O\left(\frac{K\mu\sqrt{M{\tau_{\rm{c}}}}}{{\widetilde\lambda^{(n)}}{\left|\Omega\right|}}\right)\\
    &\buildrel{(a)}\over=\left(\frac{\widehat\lambda^{(n)}-\lambda^{(n)}+ \sqrt{\mu} {\left(M \tau_{\rm{c}}\right)^{1/4}}}{\widehat\lambda^{(n)}+\sqrt{\mu} {\left(M \tau_{\rm{c}}\right)^{1/4}}}\right)\frac{\lambda^{(n)}K}{\left|\Omega\right|}+O\left(\frac{K\mu\sqrt{M{\tau_{\rm{c}}}}}{\left({\widehat\lambda^{(n)}+\sqrt{\mu} {\left(M \tau_{\rm{c}}\right)^{1/4}}}\right){\left|\Omega\right|}}\right)\\
    &=\left({\widehat\lambda^{(n)}-\lambda^{(n)}+ \sqrt{\mu} {\left(M \tau_{\rm{c}}\right)^{1/4}}}\right)\frac{\lambda^{(n)}}{\widehat\lambda^{(n)}+\sqrt{\mu} {\left(M \tau_{\rm{c}}\right)^{1/4}}}\frac{K}{\left|\Omega\right|}\\
    &+O\left(\frac{K\mu\sqrt{M{\tau_{\rm{c}}}}}{\left({\widehat\lambda^{(n)}+\sqrt{\mu} {\left(M \tau_{\rm{c}}\right)^{1/4}}}\right){\left|\Omega\right|}}\right),\\
\end{aligned}
\end{equation}
where (a) follows from the definition of ${\widetilde\lambda^{(n)}}$ in (\ref{equ_wlu}) and $\alpha=-{\lambda^{(n)} {K}}/{\left|\Omega\right|}$.

By Corollary 2.3.6 from \cite{tao2012topics}, we have $\left\|\mathbb{\widehat{W}}^{(n-1)}-\left(\mathbf{J}^{(n-1)}\right)^{\rm{H}}\mathbf{J}^{(n-1)}\right\|_2=O\left(\mu\sqrt{{M}{\tau_{\rm{c}}}}\right)$ with high probability. Recall that $\widehat{\lambda}^{(n)}$ and $\lambda^{(n)}$ are respectively the largest eigenvalues of $\mathbb{\widehat{W}}^{(n-1)}$ and $\left(\mathbf{J}^{(n-1)}\right)^{\rm{H}}\mathbf{J}^{(n-1)}$. Then we have $\left|\widehat{\lambda}^{(n)}-\lambda^{(n)}\right|=O\left(\sqrt{\mu}\left({M}{\tau_{\rm{c}}}\right)^{1/4}\right)$ according to Weyl's inequality\cite{dwork2014analyze}, which implies that
\begin{equation}
    {\widehat\lambda^{(n)}-\lambda^{(n)}+ \sqrt{\mu} {\left(M \tau_{\rm{c}}\right)^{1/4}}}=O\left(\sqrt{\mu}\left({M}{\tau_{\rm{c}}}\right)^{1/4}\right),
\end{equation}
and therefore
\begin{equation}
\frac{\lambda^{(n)}}{\widehat\lambda^{(n)}+\sqrt{\mu} {\left(M \tau_{\rm{c}}\right)^{1/4}}}=O\left(1\right).
\end{equation}
Hence, the first term in (\ref{equ_diff}) is $O\left(\frac{K}{\left|\Omega\right|}\sqrt{\mu}\left({M}{\tau_{\rm{c}}}\right)^{1/4}\right)$. Since $\widehat{\lambda}^{(n)}\ge0$, the second term in (\ref{equ_diff}) is $O\left(\frac{K}{\left|\Omega\right|}\sqrt{\mu}\left({M}{\tau_{\rm{c}}}\right)^{1/4}\right)$. In conclusion, we have
$\widehat{\alpha}-\alpha\le O\left(\frac{K}{\left|\Omega\right|}\sqrt{\mu}\left({M}{\tau_{\rm{c}}}\right)^{1/4}\right)$ with high probability.

\bibliographystyle{IEEEtran}
\bibliography{IEEEabrv,ref}
\end{document}